\documentclass[11pt]{article}
\usepackage[margin=1in]{geometry}
\usepackage[utf8]{inputenc}
\usepackage{fullpage} 
\usepackage{amsmath,amsthm,amsfonts,amssymb}
\usepackage{thm-restate}
\usepackage{graphicx}
\usepackage{colortbl}
\usepackage[pagebackref]{hyperref}
\hypersetup{
    bookmarksnumbered=true, 
    unicode=false, 
    pdfstartview={FitH}, 
    pdftitle={Title}, 
    pdfauthor={Authors}, 
    pdfsubject={}, 
    pdfcreator={}, 
    pdfproducer={}, 
    pdfkeywords={}, 
    pdfnewwindow=true, 
    colorlinks=true, 
    linkcolor=blue, 
    citecolor=blue, 
    filecolor=blue, 
    urlcolor=blue 
}
\newcommand{\eprint}[1]{\textup{\texttt{\href{http://arxiv.org/abs/#1}{#1}}}}

\renewcommand{\backref}[1]{}

\renewcommand{\backrefalt}[4]{%
\ifcase #1 %
\or
[p.\ #2]%
\else
[pp.\ #2]%
\fi}

\let\oldproofname=\proofname
\renewcommand{\proofname}
    {\upshape\bfseries{\oldproofname}}

\newtheorem{theorem}{Theorem}
\newtheorem{definition}[theorem]{Definition}
\newtheorem{corollary}[theorem]{Corollary}
\newtheorem{lemma}[theorem]{Lemma}
\newtheorem{conjecture}[theorem]{Conjecture}

\newtheorem{innerthmnum}{Theorem}
  {\renewcommand\theinnerthmnum{#1}\innerthmnum}
  {\endinnerthmnum}

\newcommand{\eq}[1]{\hyperref[eq:#1]{(\ref*{eq:#1})}}
\renewcommand{\sec}[1]{\hyperref[sec:#1]
    {Section~\ref*{sec:#1}}}
\newcommand{\thm}[1]{\hyperref[thm:#1]
    {Theorem~\ref*{thm:#1}}}
\newcommand{\lem}[1]{\hyperref[lem:#1]{Lemma~\ref*{lem:#1}}}
\newcommand{\prop}[1]{\hyperref[prop:#1]
    {Proposition~\ref*{prop:#1}}}
\newcommand{\cor}[1]{\hyperref[cor:#1]
    {Corollary~\ref*{cor:#1}}}
\newcommand{\fig}[1]{\hyperref[fig:#1]{Figure~\ref*{fig:#1}}}
\newcommand{\tab}[1]{\hyperref[tab:#1]{Table~\ref*{tab:#1}}}
\newcommand{\alg}[1]{\hyperref[alg:#1]
    {Algorithm~\ref*{alg:#1}}}
\newcommand{\app}[1]{\hyperref[app:#1]
    {Appendix~\ref*{app:#1}}}
\newcommand{\conj}[1]{\hyperref[conj:#1]
    {Conjecture~\ref*{conj:#1}}}

\newcommand{\B}{\{0,1\}}
\newcommand{\OR}{\textsc{OR}}
\newcommand{\Simon}{\textsc{Simon}}

\newcommand{\CC}{\textsc{CC}}

\DeclareMathOperator{\D}{D}
\DeclareMathOperator{\R}{R}
\DeclareMathOperator{\Q}{Q}
\DeclareMathOperator{\C}{C}
\DeclareMathOperator{\RC}{RC}
\DeclareMathOperator{\QC}{QC}
\DeclareMathOperator{\bs}{bs}
\DeclareMathOperator{\Dom}{Dom}

\DeclareMathOperator{\Hi}{H}
\DeclareMathOperator{\Bal}{Bal}
\DeclareMathOperator{\polylog}{polylog}
\DeclareMathOperator*{\E}{\mathbb{E}}
\DeclareMathOperator{\Mar}{Mar}

\title{Sculpting Quantum Speedups}
\author{
    Scott Aaronson\\
    \small\texttt{MIT}\\
    \small\texttt{aaronson@csail.mit.edu}
    \and
    Shalev Ben-David\\
    \small\texttt{MIT}\\
    \small\texttt{shalev@mit.edu}
}
\date{}

\begin{document}

\pagenumbering{gobble}

\maketitle

\begin{abstract}
Given a problem which is intractable for both
quantum and classical algorithms, can we find a sub-problem
for which quantum algorithms provide an exponential
advantage?
We refer to this problem as the ``sculpting problem.''
In this work, we give a full characterization of
sculptable functions in the query complexity setting.
We show that a total function $f$
can be restricted to a
promise $P$ such that $\Q(f|_P)=O(\polylog N)$
and $\R(f|_P)=N^{\Omega(1)}$,
if and only if $f$ has a large number of
inputs with large certificate complexity.
The proof uses some interesting techniques:
for one direction, we introduce new relationships
between randomized and quantum query complexity
in various settings, and for the other direction,
we use a recent result from communication complexity
due to Klartag and Regev.
We also characterize sculpting for
other query complexity measures, such as
$\R(f)$ vs.\ $\R_0(f)$ and $\R_0(f)$ vs.\ $\D(f)$.

Along the way, we prove some new relationships
for quantum query complexity: for example, a nearly
quadratic relationship between $\Q(f)$ and $\D(f)$
whenever the promise of $f$ is small.
This contrasts with the recent super-quadratic
query complexity separations,
showing that the maximum gap between
classical and quantum query complexities
is indeed quadratic in various settings -- just not
for total functions!

Lastly, we investigate sculpting in the Turing machine
model. We show that if there is any
$\mathsf{BPP}$-bi-immune
language in $\mathsf{BQP}$, then \emph{every}
language outside
$\mathsf{BPP}$ can be restricted to a promise which
places it in $\mathsf{PromiseBQP}$ but not in
$\mathsf{PromiseBPP}$.
Under a weaker assumption,
that some problem in $\mathsf{BQP}$
is hard on average for $\mathsf{P/poly}$,
we show that every
\emph{paddable} language outside $\mathsf{BPP}$
is sculptable in this way.
\end{abstract}

\newpage

\pagenumbering{arabic}

\section{Introduction}
When are quantum algorithms useful?
In general, quantum algorithms are believed to provide
exponential speedups for certain structured problems,
such as factoring \cite{Sho97}, but
not for unstructured problems like $\mathsf{NP}$-complete
problems.

In this work, we ask the question in a new way.
Given a problem for which quantum algorithms are not
useful, can we nevertheless find a sub-problem on which
they provide an exponential advantage over classical
algorithms?
We call this the ``sculpting'' question:
our goal is to sculpt the original intractable problem
into a sub-problem that's still classically intractable,
but for which there exists a fast quantum algorithm.
The sculpting question arises, for example, in adiabatic
quantum computation: while it is not believed that
adiabatic quantum computing can solve \textsc{NP}-complete
problems in polynomial time, a widely discussed question
is whether there is a sub-problem of \textsc{SAT} on which
adiabatic computing provides an exponential advantage
over classical algorithms.

We study the sculpting question primarily in the query
complexity model.
The utility of the model comes from its relative
tractability: for example, in query complexity,
Shor's period finding algorithm
provides a provable exponential
speedup over any classical algorithm
\cite{Cle04}.

In query complexity, we're given a
(possibly partial) function $f:\B^N\to\B$
and an oracle access to a string $x\in\B^N$.
The goal is to evaluate $f(x)$ using
as few oracle calls to the entries of $x$ as possible.
The minimum number of queries required by an algorithm
for computing $f(x)$ (over the worst-case choice of $x$)
is the query complexity of $f$. If the algorithm
in question is deterministic, we denote this by
$\D(f)$; if it is (bounded error) randomized,
we denote this by
$\R(f)$; and if it is (bounded error) quantum,
we denote it by $\Q(f)$.

In this query complexity setting,
the sculpting question can be phrased as follows:
given a total function $f:\B^N\to\B$
for which $\R(f)$ and $\Q(f)$ are both large
(say, $N^{\Omega(1)}$), is there
a promise $P\subseteq\B^N$ such that
$f|_P$, the restriction of $f$ to $P$,
has $\Q(f|_P)=O(\polylog N)$ and $\R(f|_P)=N^{\Omega(1)}$?

For example, if $f$ is the $\OR$ function,
such sculpting is not possible, as follows from
\cite{Aar06}. As another example, if $f$ is defined
to be $1$ when Simon's condition is satisfied and $0$
otherwise, then sculpting is possible: the promise
will simply restrict to inputs that either satisfy
Simon's condition or are far from satisfying it;
this promise suffices for an exponential quantum speedup
\cite{BFNR08}.

We fully characterize the functions $f$
for which such a promise exists. In particular,
we show that sufficiently ``rich'' functions, such as
\textsc{Parity} or \textsc{Majority}, are sculptable.
Our sculpting construction
uses communication complexity in a novel way. In
the other direction, to prove non-sculptability, we
prove new query complexity relationships.
As a corollary, we get nearly quadratic relationships between
classical and quantum query complexities for
a wider class of functions than previously known.

\subsection*{Results}

\subsubsection*{$\Hi$-indices}

We introduce a new query complexity measure,
$\Hi(\C_f)$, defined as the maximum number $h$
for which there are $2^h$ inputs to $f$
with certificate complexity at least $h$.
We call this the H-index of certificate
complexity (motivated by the citation H-index sometimes
used to measure research productivity \cite{H05}).
This quantity
measures the number of inputs there are to a function
$f$ that have large certificate complexity.
We prove various properties of $\Hi(\C_f)$;
most notably, we show that
for total functions, it is nearly quadratically
related to $\Hi(\bs_f)$, the H-index of block
sensitivity. This is analogous to the quadratic
relationship between $\C$ and $\bs$.

\subsubsection*{Sculpting in Query Complexity}

Our main result is the following theorem, which
neatly characterizes sculptability in the query complexity
model in terms of the H-index of certificate
complexity.

\begin{theorem}\label{thm:main0}
Let $f:\B^N\to\B$ be a total function.
Then there is a promise $P\subseteq\B^N$ such that
$\R(f|_P)=N^{\Omega(1)}$ and $\Q(f|_P)=N^{o(1)}$,
if and only if $\Hi(\C_f)=N^{\Omega(1)}$.
Furthermore, in this case we also have
$\Q(f|_P)=O(\log^2 N)$.
\end{theorem}

This theorem follows as an immediate corollary
of the following more general characterization theorem.

\begin{restatable}{theorem}{main}\label{thm:main}
For all total functions $f:\B^N\to\B$ and
all promises $P\subseteq\B^N$, we have
\[\R(f|_P)=O(\Q(f|_P)^2\Hi(\C_f)^2).\]
Conversely, for all total functions
$f:\B^N\to\B$, there is
a promise $P\subseteq\B^N$ such that
\[\R(f|_P)=\Omega\left(\frac{\Hi(\C_f)^{1/6}}
                {\log^{13/6} N}\right)\qquad
                \mbox{and}\qquad
\Q(f|_P)=O(\log^2\Hi(\C_f)).\]
\end{restatable}

We also prove an analogous theorem
for $\D$ vs.\ $\R_0$,
showing that the same $\Hi(\C_f)=N^{\Omega(1)}$
condition also characterizes sculpting $\D(f)$
vs.\ $\R_0(f)$. On the other hand, we show
that sculpting $\R_0(f)$ vs.\ $\R(f)$ is \emph{always}
possible: for every total function $f$
with $\R_0(f)=N^{\Omega(1)}$, there is a promise $P$
such that $\R_0(f|_P)=N^{\Omega(1)}$ and
$\R(f|_P)=O(1)$.


\subsubsection*{Query Complexity on Small Promises}

On the way to proving \thm{main}, we
prove the following theorem, providing
a quadratic relationship between $\Q(f)$ and $\D(f)$
when the domain of $f$ is small.
This provides a ironic twist to the query complexity
story: for a long time, it was believed that $\D(f)$
and $\Q(f)$ are quadratically related when the domain
of $f$ is very large (in particular, for total functions).
This conjecture was recently disproven by
\cite{ABB+15} (who showed a $\D(f)\sim\Q(f)^4$ separation)
and by \cite{ABK15}
(who showed an $f$ such that $\R(f)\sim\Q(f)^{2.5}$).
Instead, we now show
that the quadratic relationship holds when the domain of
$f$ is very \emph{small}.

\begin{restatable}{theorem}{small}\label{thm:small}
Let $f:\B^N\to\B$ be a partial function, and let
$\Dom(f)$ denote the domain of $f$. Then
\[\Q(f)=
\Omega\left(\frac{\sqrt{\D(f)}}{\log|\Dom(f)|}\right).\]
\end{restatable}

\subsubsection*{Query Complexity for Unbalanced Functions}

We show two relationships similar to \thm{small}
that hold for functions whose domain is large,
but which are unbalanced: they contain very few
$0$-inputs compared to $1$-inputs, or vice versa.

\begin{restatable}{theorem}{bal}\label{thm:bal}
Let $f:\B^N\to\B$ be a partial function.
Define the measure $\Bal(f)\in[0,N]$ as $\Bal(f):=1+\min\{\log|f^{-1}(0)|,\log|f^{-1}(1)|\}$
(or $0$ if $f$ is constant). Then
\[\R(f)=O(\Q(f)^2\Bal(f))\]
\[\D(f)=O(\R_0(f)\Bal(f)).\]
A similar polynomial relationship between $\R_0(f)$
and $\R(f)$ does not hold in general.
\end{restatable}

\subsubsection*{New Relationship for Total Functions}

We prove the following new query complexity relationship
for total functions, generalizing the known relationship
$\D(f)=O(\Q(f)^2\C(f))$.

\begin{restatable}{theorem}{tot}\label{thm:tot}
Let $f:\B^N\to\B$ be a total function. Then
\[\D(f)=O(\Q(f)^2\Hi(\sqrt{\C_f})^2).\]
\end{restatable}

Here $\Hi(\sqrt{\C_f})$ denotes the H-index
of the square root of certificate complexity; this
is the maximum number $h$ such that there are at least
$2^h$ inputs to $f$ for which $\sqrt{\C(f)}$
is at least $h$.
We note that $\Hi(\sqrt{\C_f})^2\leq\C(f)$
for all total functions, so this is an improvement
over the relationship $\D(f)=O(\Q(f)^2\C(f))$.
Moreover, when $f=\OR$, we have $\Hi(\sqrt{\C_f})^2=1$
and $\C(f)=N$, so this improvement is strict.

We remark that this result could let us improve the
relationship $\D(f)=O(\Q(f)^6)$ if we could show
$\Hi(\sqrt{\C_f})^2=o(\Q(f)^4)$.
\thm{tot} therefore provides a new approach for this
long-standing open problem.


\subsubsection*{Sculpting in the Turing Machine Model}

In \sec{TM}, we examine sculpting in the
Turing machine model. We say that a language $L$
is \emph{sculptable} if there is a promise set $P$
such that the promise problem of deciding if an input
from $P$ is in $L$ is in $\mathsf{P{}romiseBQP}$
but not in $\mathsf{P{}romiseBPP}$.
We prove two sculptability theorems, both of them
providing evidence that most or all languages
outside of $\mathsf{BPP}$ are sculptable.

\begin{restatable}{theorem}{pad}\label{thm:pad}
Assume $\mathsf{P{}romiseBQP}$ is hard on average
for $\mathsf{P/poly}$. 
Then every paddable language outside of $\mathsf{BPP}$
is sculptable.
\end{restatable}

\begin{restatable}{theorem}{immune}\label{thm:immune}
Assume there exists a $\mathsf{BPP}$-bi-immune
language in $\mathsf{BQP}$. Then every language
outside of $\mathsf{BPP}$ is sculptable.
\end{restatable}

For the definitions of paddability and bi-immunity,
see \sec{TM}. These theorems assume very little about
$\mathsf{BQP}$ and $\mathsf{BPP}$, and analogous statements
hold for other pairs of complexity classes.


\section{Preliminaries}
\label{sec:prelim}

For a (possibly partial) function $f:\B^N\to\B$,
we use $\D(f)$, $\R_0(f)$, $\R(f)$, and $\Q(f)$
to denote the deterministic query complexity,
zero-error randomized query complexity,
bounded-error randomized query complexity,
and bounded-error quantum query complexity of $f$,
respectively. For the definitions of these measures,
see \cite{BdW02}.

A partial assignment is a string $p$ in $\{0,1,*\}^N$
that represents partial knowledge of a string in $\B^N$.
For $x\in\B^N$, we say that $p$ is a partial assignment
of $x$ if $x$ extends $p$; that is, if $x$ and $p$ agree
on all the non-$*$ entries of $p$. A partial assignment
$p$ for $x$ is called a \emph{certificate} for $x$ if all
strings that extend $p$ have that same value under $f$
as $x$; that is, if for all $y\in\B^N$ that extend
$p$, we have $f(y)=f(x)$.
The certificate
complexity of $f$ on input $x$, denoted by $\C_f(x)$,
is the minimum number of bits in any certificate of $x$
with respect to $f$. The certificate complexity of $f$,
denoted by $\C(f)$, is defined as the maximum of
$\C_f(x)$ over all strings $x$ in the domain of $f$.

The certificate complexity $\C_f(x)$ can be thought of
as the deterministic query complexity of $f$ given
the promise that the input is either $x$ or else
an input $y$ such that $f(x)\neq f(y)$.
Motivated by this observation, Aaronson \cite{Aar06} defined the
\emph{randomized certificate complexity} of $x$, denoted
by $\RC_f(x)$, to be the 
(bounded-error) randomized query complexity
of $f$ on this promise. He defined the quantum
certificate complexity $\QC_f(x)$ analogously.
As with $\C$, we use $\RC(f)$ to denote the maximum
of $\RC_f(x)$ over all $x$ in the domain of $f$,
and define $\QC(f)$ similarly.

For any string $x\in\B^N$ and set of bits $B$,
denote by $x^B$ the string $x$ with the bits in $B$
flipped. For any $f:\B^N\to\B$, if $f(x)\neq f(x^B)$,
we say that $B$ is a \emph{sensitive block} for $x$
with respect to $f$. The \emph{block sensitivity}
of $x$, denoted by $\bs_f(x)$, is the maximum number
of disjoint sensitive blocks for $x$.
Note that the block sensitivity is the packing number
of the collection of sensitive blocks of $x$.
It can be seen that $\C_f(x)$ can be interpreted
as the hitting number of that collection
(the minimum number of bits required to hit all the blocks).
Moreover, $\RC_f(x)$ is simply the fractional packing number
(which equals the fractional hitting number).
In other words, $\RC_f(x)$ can be interpreted as the
fractional block sensitivity (or fractional certificate
complexity). These observations are implicit in
\cite{Aar06}, and were made explicit in \cite{KT13}.

Actually, the fractional block sensitivity
differs by a constant factor
from Aaronson's original definition of $\RC$.
In this work we will use $\RC$ to denote the fractional
block sensitivity.
Another property of $\RC$ that we will need is
that $1/\RC_f(x)$ is equal
to the minimum infinity-norm distance between $x$
and the convex hull of the set of inputs $y$
such that $f(y)\neq f(x)$. That is, if $f(x)=0$
and $S=f^{-1}(1)$, we have
\[\frac{1}{\RC_f(x)}=\min_{\mu\in\Delta_{S}}\max_i
    \Pr_{y\sim \mu}[x_i\neq y_i],\]
where $\Delta_S$ is the set of all probability distributions
over $S$ (equivalently, the convex hull of $S$).
In particular, this minimum is attained, so there is
a probability distribution $\mu$ over $f^{-1}(1)$
such that for all $i=1,2,\ldots,n$, if we sample
$y\sim\mu$ we get $\Pr[y_i\neq x_i]\leq 1/\RC_f(x)$.

Clearly, for all $f:\B^N\to\B$ we have
\[\left\{\QC(f),\bs(f)\right\}\leq\RC(f)\leq
    \left\{\C(f),\R(f)\right\}\leq \R_0(f)\leq \D(f),\]
with $\Q(f)$ lying between $\QC(f)$ and $\R(f)$.
Aaronson \cite{Aar06} showed that
$\RC_f(x)=\Theta(\QC_f(x)^2)$ for all $f$ and $x$,
so in particular, $\RC(f)=\Theta(\QC(f)^2)$.
In addition, 
when $f$ is total, we can relate all these measures to
each other: we have
\[\D(f)=O(\bs(f)^3)=O(\RC(f)^3)=O(\QC(f)^6)=O(\Q(f)^6).\]

%
%

\subsection*{Balance and H Indices}

We will use $\Dom(f)$ to denote the domain of a partial
function $f$. We define $\Bal(f)$ to be $0$
if $f$ is constant, and otherwise,
to be the minimum of $1+\log|f^{-1}(0)|$ and
$1+\log|f^{-1}(1)|$ (we use $\log$ to denote logarithm
base $2$). Note that since 
$|f^{-1}(0)|+|f^{-1}(1)|=|\Dom(f)|\leq 2^N$,
we have $\Bal(f)\leq N$. Thus $\Bal(f)\in[0,N]$.

We will use a new set of query complexity
measures called H-indices (the name is motivated
by the H-index measure of citations, a common
metric for research output). For
a given function $g:\B^N\to[0,\infty)$,
we will define the H-index of $g$,
denoted by $\Hi(g)$, as the maximum
number $h$ such that there are at least $2^h$ inputs
with $g(x)\geq h$. Alternatively, the H-index
of $g$ can be defined as the minimum number $h$
such that there are at most $2^h$ inputs with $g(x)>h$.
It is not obvious that these definitions
are equivalent (or even that the minimum and
maximum are attained); we prove this in \app{H}.

Note that $\Hi(g)\in[0,N]$, and
$\Hi(g)\leq\max_x g(x)$. Also, if $g(x)\geq g^\prime(x)$
for all $x\in\B^N$, we have
$\Hi(g)\geq\Hi(g^\prime)$.

We'll primarily
be interested in measures like $\Hi(\C_f)$,
$\Hi(\RC_f)$, and $\Hi(\bs_f)$. We have
$\Hi(\C_f)\leq \C(f)$, $\Hi(\RC_f)\leq\RC(f)$,
and $\Hi(\bs_f)\leq\bs(f)$.
We also have
\[\Hi(\bs_f)\leq\Hi(\RC_f)\leq\Hi(C_f).\]
The H-index of certificate complexity can be much
smaller than the certificate complexity itself.
For example, the $\OR$ function has only one certificate
of size greater than $1$, so
$\Hi(\C_{\OR})=1$, even though $\C(\OR)=n$.

In \app{H}
we show that if $\alpha:[0,\infty)\to[0,\infty)$
is an increasing function, then
\[\Hi(\alpha\circ g)\leq\max\{\Hi(g),\alpha(\Hi(g))\}.\]
In particular, this will imply $\Hi(\C_f^2)\leq\Hi(\C_f)^2$.

\subsection*{Shattering and the Sauer-Shelah Lemma}

For a set of indices $A\subseteq\{1,2,\ldots,N\}$,
let $S|_A\subseteq \B^{|A|}$ be the set of restrictions
of each string in $S$ to the indices in $A$.
We say $A$ is shattered by $S$ if $S|_A=\B^{|A|}$.
In other words, $A$ is shattered by $S$
if $S$ has all possible behaviors on $A$.
The Sauer-Shelah lemma \cite{Sau72,She72}
is a classic result that upper-bounds the size of $S$
in terms of the size of $A$. We will use
the following corollary of it.

\begin{restatable}{lemma}{sauer}
\label{lem:sauer}
Let $S\subseteq\B^N$ be a collection of strings.
Then there is a shattered set of indices
of size at least
\[\frac{\log |S|}{\log (N+1)}.\]
\end{restatable}

\lem{sauer} follows straightforwardly from the
Sauer-Shelah lemma, as we prove in \app{sauer}.
We will often use the weaker bound
$\frac{\log |S|}{2\log N}$ instead, which holds for
$N\geq 2$. This will sometimes lead to simpler formulas.

\section{Non-Sculptability Theorems}
\label{sec:nosculpt}

In this section, we prove the non-sculptability direction
of \thm{main}. The proof has two parts: in
\sec{bal}, we prove a relationship between randomized
and quantum query complexities for
``unbalanced'' functions,
and in \sec{H}, we use this to prove a sculpting
lower bound in terms of the H-index of
certificate complexity.

\subsection{Query Complexity for Unbalanced Functions}
\label{sec:bal}

We wish to show a nearly-quadratic relationship
between randomized and quantum query complexities
for functions $f$ for which $\Bal(f)$ is small.
Note that this is a generalization of the relation
$\RC_f(x)=O(\Q_f(x)^2)$ from \cite{Aar06}.
That is, \cite{Aar06} showed that
for the task of distinguishing one input from
a (possibly large) set of alternatives, randomized
and quantum algorithms are quadratically related.
We want a similar relationship for the task
of distinguishing a \emph{small set} of inputs
from a (possibly large) set of alternatives.

We start with the following lemma.

\begin{lemma}\label{lem:RvsQlemma}
Let $f:\B^N\to\B$ be a partial function.
For $a\notin f^{-1}(0)$, let $f_{a,0}$ be the problem
of distinguishing $a$ from $f^{-1}(0)$.
That is, $f_{a,0}$ is the function
$f_{a,0}:\{a\}\cup f^{-1}(0)\to\{0,1\}$
with $f(x)=1$ iff $x=a$.
For $a\notin f^{-1}(1)$, define $f_{a,1}$ analogously.
Then for all $a\in\{0,1\}^N$, we have either
$\R(f_{a,0})=O(\Q(f)^2)$ or $\R(f_{a,1})=O(\Q(f)^2)$.

Note that this holds even when $a$ is not
in the promise of $f$. The constant in the big-$O$
notation is a universal constant independent of $a$, $f$,
and $N$.
\end{lemma}

\begin{proof}
Let $Q$ be the quantum algorithm that achieves $\Q(f)$
quantum query complexity in determining the value of
$f$ on a given input. When run on any $a\in f^{-1}(0)$,
$Q$ will output $0$ with probability at least $2/3$,
and when run on $a\in f^{-1}(1)$, it will output
$1$ with probability at least $2/3$.

Consider running $Q$ on an input $a\notin \Dom(f)$.
Then $Q$ will output $0$ with some probability $p$
and output $1$ with probability $1-p$.
If $p\geq 1/2$, then $Q$ distinguishes $a$ from
$f^{-1}(1)$ with constant probability.
If $p\leq 1/2$, then $Q$ distinguishes $a$ from
$f^{-1}(0)$ with constant probability.
Thus for all $a\in\{0,1\}^N$, we have either
$\Q(f_{a,0})=O(\Q(f))$ or $\Q(f_{a,1})=O(\Q(f))$.
From \cite{Aar06},
we have $\RC(g)=O(\QC(g)^2)=O(\Q(g)^2)$
for all functions $g$,
so we conclude that either $\RC(f_{a,0})=O(\Q(f)^2)$
or $\RC(f_{a,1})=O(\Q(f)^2)$.

Finally, note that for a problem of
distinguishing one input from the rest,
randomized query complexity equals
randomized certificate complexity.
Thus we get that
for all $a\in\{0,1\}^N$, either
$\R(f_{a,0})=O(\Q(f)^2)$ and
or $\R(f_{a,1})=O(\Q(f)^2)$.
\end{proof}

We're now ready to prove the desired relationship
between $\R$ and $\Q$.

\begin{theorem}\label{thm:RvsQ}
    Let $f:\B^N\to\{0,1\}$ be a partial function. Then
    \[\R(f)=O(\Q(f)^2\Bal(f)).\]
\end{theorem}

\begin{proof}
Without loss of generality, assume
$|f^{-1}(0)|\leq|f^{-1}(1)|$.
We use \lem{RvsQlemma} to construct a
randomized algorithm
for determining $f(x)$ given oracle access to $x$,
assuming that $f^{-1}(0)$ is small.
The idea is to keep track of the subset
$Z\subseteq f^{-1}(0)$
of strings that the input $x$ might feasibly be
(consistent with the queries seen so far).
We then construct a string $a$ from a majority
vote of the elements of $Z$; that is,
for each index $i\in [n]$, $a_i$ will be the majority
of $y_i$ over all $y\in Z$
(with ties broken arbitrarily).

This string $a$ need not be in $\Dom(f)$.
The important property of it is that
if we query an index $i$ of the input $x$
and discover that $x_i\neq a_i$,
we can eliminate at least half of the strings from $Z$,
since they are no longer feasible
possibilities for $x$.

We then get the following randomized
algorithm for evaluating $f(x)$:

\begin{itemize}
    \item Initialize $Z=f^{-1}(0)$.
    \item While $Z\neq\emptyset$:
    \begin{enumerate}
        \item Calculate $a$ from the entry-wise
                majority vote of $Z$.
        \item Determine $b\in\{0,1\}$ such that
                $\R(f_{a,b})=O(\Q(f)^2)$
                (this exists by \lem{RvsQlemma}).
        \item Run the randomized
                algorithm evaluating $f_{a,b}$
                on $x$ with some amplification\\
                (to be specified later).
        \item If its output is $1$ (i.e.\
                the algorithm thinks $x=a$
                rather than $x\in f^{-1}(b)$),\\
                output $1-b$ and halt.
        \item If its output is $0$,
                a bit $i$ was queried
                to reveal $x_i\neq a_i$, so update $Z$\\
                (removing at least half its elements).
    \end{enumerate}
    \item
    If $Z=\emptyset$, output $1$.
\end{itemize}

We note a few things about this algorithm.
First, in step 3, notice that $x$ need
not be in the domain of $f_{a,b}$.
However, we may still run the randomized
algorithm that evaluates $f_{a,b}$,
and use the fact that if $x$
does happen to be in the domain (in particular,
if $x\in f^{-1}(b)$), then the algorithm will
work correctly. This is exactly what we use
in step 4: if the algorithm that distinguishes
$a$ from $f^{-1}(b)$ says that $x$ is equal to $a$,
it need not mean that $x$ is in fact equal to $a$,
but it does mean that $x\notin f^{-1}(b)$.

Secondly, step 5 assumes that the randomized algorithm
for evaluating $f_{a,b}$
will only conclude that an input
$x$ is not equal to $a$
if it finds a disagreement with $a$.
This is a safe assumption,
as argued in Lemma 5 of \cite{Aar06}.

Finally, we determine the number of queries this
algorithm uses.
The outer loop happens at most
$\lfloor\log|f^{-1}(0)|\rfloor+1\leq\Bal(f)$ times.
Step 3 in the loop is the only one which
queries the input string.
Since the loop repeats at most $\Bal(f)$ times,
we can safely amplify the algorithm in step 3
$O(\log\Bal(f))$ times.
This gives a query complexity of
$O(Q(f)^2\log\Bal(f))$
for step 3, so the overall number of queries is
$O(Q(f)^2\Bal(f)\log\Bal(f))$.

We can get rid of the log factor by being more careful
with the amplification. Note that if we ever find a
disagreement with $a$ when running the algorithm,
we may immediately stop
amplifying and proceed to step 5.
We keep a count $c_0$ of how
many times we had to amplify in step 3
for functions of the
form $f_{a,0}$, and a count
$c_1$ for functions of the form $f_{a,1}$.

If $c_0$ ever reaches $2\Bal(f)$,
we output $1$ and halt.
Similarly, if $c_1$ ever reaches $2\Bal(f)$,
we output $0$ and halt.
This ensures the total amplification
is $O(\Bal(f))$,
so the total query complexity of the
algorithm is $O(Q(f)^2\Bal(f))$.

Note that if $f(x)=0$
and the output of the algorithm was $1$,
it means that we ran the algorithm evaluating $f_{a,0}$
(for varying values of $a$)
$2\Bal(f)$ times, and at most $\Bal(f)$
of those times the algorithm said that $x\in f^{-1}(0)$.
For each individual run,
the probability is at least $2/3$
that the algorithm would say that $x\in f^{-1}(0)$.
An application of the
Chernoff bound shows that the probability
of this happening is exponentially small.
Similarly, the probability of the algorithm giving $0$
when in actuality $f(x)=1$ is also exponentially small.

We conclude that $R(f)=O(Q(f)^2\Bal(f))$, as desired.
\end{proof}

\subsection{Application to Non-Sculptability}
\label{sec:H}

\thm{RvsQ} immediately gives the following
non-sculptability result, which says that
unbalanced functions cannot be sculpted.

\begin{corollary}\label{cor:bal_lower}
Let $f:\B^N\to\B$ be a total function.
For any promise $P\subseteq\B^N$, we have
\[\R(f|_P)=O(\Q(f|_P)^2\Bal(f)).\]
\end{corollary}

\begin{proof}
Note that $\Bal(f|_P)\leq\Bal(f)$
for any $f$ and $P$. Then, by \thm{RvsQ},
we have
\[\R(f|_P)=O(\Q(f|_P)^2\Bal(f|_P))=O(\Q(f|_P)^2\Bal(f)).\]
\end{proof}

We extend this result by showing
that any function with a small number of large
certificates also cannot be sculpted.
This gives us a non-sculptability result in terms
of the H-index of certificate complexity.

\begin{theorem}\label{thm:noSculpting}
    Let $f:\B^N\to\{0,1\}$ be a total function.
    For any promise $P\subseteq\B^N$, we have
    \[\R(f|_P)=O(\Q(f|_P)^2\Hi(\C_f^2)).\]
\end{theorem}

\begin{proof}
We design a deterministic algorithm
that reduces the set of possibilities for the input
to an unbalanced set.
Specifically, the algorithm will reduce the
possibilities for the input to a set $S\subseteq\B^N$
such that $\Bal(f|_S)\leq\Hi(\C_f^2)+1$.
We then use \thm{RvsQ}
to get the desired non-sculptability result.

Note that every $1$-certificate of $f$ must conflict
with every $0$-certificate of $f$ in at least one bit.
Therefore, by querying all non-$*$
entries of a $0$-certificate,
we reveal at least one entry of each $1$-certificate.

We design a deterministic algorithm
for computing $f$ on an input
from $P$. The algorithm proceeds as follows:
it repeatedly picks
a $0$-certificate $p$ for $f$ of size at most
$\sqrt{\Hi(\C_f^2)}$
that is consistent with all the entries of the input that
were revealed so far. It then queries all the non-$*$
entries of $p$. This is repeated
$\sqrt{\Hi(\C_f^2)}$ times, or until
there are no $0$-certificates of size at most
$\sqrt{\Hi(\C_f^2)}$
(whichever happens first). Finally, the algorithm
returns the set $S$ of strings that are consistent with
the revealed entries of the input.

This algorithm uses at most $\Hi(\C_f^2)$ queries.
We check its correctness by examining the set $S$.
Clearly, the input is in $S$. Furthermore, if any certificate
of $f$ was revealed, then $f$ is constant on $S$,
so $S$ contains either no $0$-inputs or no $1$-inputs.

There are at most $2^{\Hi(\C_f^2)}$ inputs with
certificate complexity larger than $\sqrt{\Hi(\C_f^2)}$.

If the algorithm terminated because there were no
consistent $0$-certificates, then the only
$0$-inputs in $S$ have certificates of size larger than
$\sqrt{\Hi(\C_f^2)}$. There are at most
$2^{\Hi(\C_f^2)}$ of them,
so $S$ has at most $2^{\Hi(\C_f^2)}$ $0$-inputs to $f$.
Conversely, if the algorithm went through
$\sqrt{\Hi(\C_f^2)}$ iterations
of querying consistent $0$-certificates, then
it must have revealed $\sqrt{\Hi(\C_f^2)}$
entries of each $1$-certificate
to $f$. If no $1$-certificate was discovered, it means
the revealed entries contradicted all $1$-certificates
of size at most $\sqrt{\Hi(\C_f^2)}$.
Thus the only $1$-inputs in $S$
have certificate size greater than $\sqrt{\Hi(\C_f^2)}$,
from which
it follows that there are less than $2^{\Hi(\C_f^2)}$
of them.

We conclude that $S$ contains either at most
$2^{\Hi(\C_f^2)}$
$0$-inputs to $f$ or at most $2^{\Hi(\C_f^2)}$
$1$-inputs to $f$. This gives
$\Bal(f|_S)\leq \Hi(\C_f^2)+1$.

We design a randomized algorithm for $f|_P$ as
follows. First, we run the above deterministic
algorithm to reduce the problem of computing
$f|_P$ to the problem of computing $f|_{S\cap P}$.
This costs $\Hi(\C_f^2)$ queries.
By \thm{RvsQ}, there is a randomized algorithm
that uses
\[O(\Q(f|_{S\cap P})^2\Bal(f|_{S\cap P}))
    =O(\Q(f|_P)^2\Bal(f|_S))
    =O(\Q(f|_P)^2\Hi(\C_f^2))\]
queries to compute $f|_{S\cap P}$. Running
this algorithm allows us to compute $f|_P$.
The total number of queries used was
\[O(\Q(f|_P)^2\Hi(\C_f^2)+\Hi(\C_f^2))
=O(\Q(f|_P)^2\Hi(\C_f^2)).\]
\end{proof}

Note that \thm{noSculpting} completes the first
part of the proof of \thm{main},
since $\Hi(\C_f^2)\leq\Hi(\C_f)^2$.
It is natural to wonder whether \thm{noSculpting}
is always at least as strong as \cor{bal_lower}.
In \thm{Hbal}, we will show that it is,
up to a quadratic factor and a $\log N$ factor.

\section{Sculpting from Communication Complexity}
\label{sec:sculpt}

In this section, we show that if a function $f$
has many inputs with large randomized certificate
complexity then it \emph{can} be sculpted:
there is a promise $P$ so that $f|_P$ exhibits
a large quantum speedup. This means that if $\Hi(\RC_f)$
is large, the function $f$ can be sculpted.
In \sec{Hprops}, we will relate $\Hi(\RC_f)$
to $\Hi(\C_f)$, thereby completing the proof of \thm{main}.

Our sculptability proof will rely
on the solution to a problem we call the ``extended queries
problem,'' which might be of independent interest.
The solution to this problem will in turn
use results from communication complexity.

\subsection{The Extended Queries Problem}

We usually let an algorithm for computing
a (possibly partial) function
$f:\B^N\to\B$  query the bits of the input $x$.
But what happens if we let the algorithm
make other types of queries? For example,
if $x$ is a Boolean string, we can let the algorithm
query the parity of $x$. How does this extra power
affect the query complexity of $f$?
In particular, is there some special set of additional
queries such that if a randomized algorithm is allowed
to make the special queries, it can
simulate \emph{any} quantum algorithm?
If so, how many special queries suffice for this property to
hold?

To formalize this question, we need a few definitions.

\begin{definition}
An \emph{extension function with extension $G$}
is an injective total function
$\phi:\B^N\to\B^G$ (in particular, we need
$G\geq N$).
\end{definition}

An extension function
specifies, for each input $x\in\B^N$, the types of
queries an algorithm is allowed to make on $x$.
In other words, we will let algorithms query from
$\phi(x)$ instead of from $x$.
Note that the extension function may provide easy access
to information about $x$ that is hard to obtain
otherwise (such as its parity).

\begin{definition}
Let $f:\B^N\to\B$ be a partial function, and let
$\phi$ be an extension function.
The \emph{extended version
of $f$ with respect to $\phi$} is the partial function
$f^{\phi}:\phi(\Dom(f))\to\B$ defined by
$f^{\phi}(x)=f(\phi^{-1}(x))$.
\end{definition}

Note that
$f^{\phi}$ is a partial function from $\B^G$ to $\B$.
We can consider $\D(f^{\phi})$, $\R(f^{\phi})$,
$\Q(f^{\phi})$, and so on.
To pose the extended queries problem, we will need a notion
of the complexity of a set of functions, defined as the
maximum complexity of any function in that set.

\begin{definition}
For any set of functions $S$, we define
$\D(S):=\max_{f\in S} \D(f)$.
We define $\R(S)$, $\Q(S)$, etc.\ similarly.
Further, we
define $\D^G(S)$, the \emph{extended query complexity
of $S$ with extension $G$}, to be the minimum, over all
extension functions $\phi$ with extension $G$,
of $\max_{f\in S}\D(f^{\phi})$.
We define $\R^G(S)$, $\Q^G(S)$, etc.\ similarly.
\end{definition}

In other words, for any set of functions, the extended query
complexity of the set with $G$ extension is the number
of queries required to compute all functions in the
set given the best possible extension. We observe
that if $G\geq |S|$, the extended query complexity
$\D^G(S)$ is $1$, since the extension $\phi(x)$
for a given input $x$ could
simply specify the values of all the functions in $S$
on $x$. We also observe that for all $G\geq N$, we have
$\D^G(S)\leq \D(S)$, since the identity function
is always a valid extension function.
Moreover, the extended query complexity of a set
is decreasing in $G$. We now ask the following question.

\textbf{The Extended Queries Problem.}
Is there a set of functions
$S$ for which $\Q(S)$ is small but $\R^G(S)$ is large,
even when the extension $G$ is exponentially large
in the input size $N$?
We can also ask this question for other complexity
measures, such as $\R(S)$ vs. $\R_0^G(S)$ or
$\R_0(S)$ vs.\ $\D^G(S)$.

It turns out that a
negative solution to the extended queries problem
implies a sculptability result in terms of $\Hi(\RC_f)$,
as the following theorem shows.

\begin{restatable}{theorem}{sculptingExtension}
\label{thm:sculpting_extension}
    Let $f:\B^N\to\B$ be a total function.
    Let $A=\frac{\Hi(\RC_f)}{4\log N}$,
    and let $S$ be any set of partial functions
    from $\B^A$ to $\B$. Then there is a promise
    $P\subseteq\B^N$ such that
    \[\Q(f|_P)=O(\Q(S)),\qquad
        \R(f|_P)=\Omega(\R^N(S)).\]
    Analogous statements hold for other pairs of
    complexity measures, such as $\D$ and $\R_0$ or
    $\R_0$ and $\R$.
\end{restatable}

We delay the proof of \thm{sculpting_extension}
to \sec{sculpting_extension}.
First, we settle the extended queries problem for
$\R$ vs.\ $\Q$:
\thm{communication_extension}
will provide an exponential lower bound
on $G$ by reducing the extended queries problem
to a problem in communication complexity.

\subsection{Reducing Extension to Communication Complexity}

For a partial function $f:\B^{N_1}\times\B^{N_2}\to\B$,
we will denote the communication complexities of $f$
by $\D^{\CC}(f)$, $\R^{\CC}(f)$, $\Q^{\CC}(f)$, and
$\R^{\CC}_0(f)$. We will use the following
definition.

\begin{definition}
Let $f:\B^{N_1}\times\B^{N_2}\to\B$ be a partial function.
For any $x\in\Dom(f)$, we write $x=x_Ax_B$, with
$x_A\in\B^{N_1}$ and $x_B\in\B^{N_2}$.
Let $\Dom_A(f)=\{x_A:x\in\Dom(f)\}$
and $\Dom_B(f)=\{x_B:x\in\Dom(f)\}$.
For any $a\in\Dom_A(f)$, we define the
marginal of $f$ with
respect to $a$ to be the partial
function $f_a:\B^{N_2}\to\B$
defined by $f_a(b):=f(a,b)$ for all $b\in\B^{N_2}$
such that $(a,b)\in\Dom(f)$.
We define $\Mar(f)=\{f_a:a\in\Dom_A(f)\}$
to be the set of all marginal functions for $f$.
\end{definition}

We now connect communication complexity to the extended
queries problem.

\begin{theorem}\label{thm:communication_extension}
    Let $f:\B^{N_1}\times\B^{N_2}\to\B$
    be a partial function. Then
    \[\R^G(\Mar(f))=\Omega\left( \frac{\R^{\CC}(f)}
        {\log G}\right).\]
    Similarly, we also have
    $\D^G(\Mar(f))=\Omega(\D^{\CC}(f)/\log G)$,
    $\R_0^G(\Mar(f))=\Omega(\R_0^{\CC}(f)/\log G)$,
    and
    $\Q^G(\Mar(f))=\Omega(\Q^{\CC}(f)/\log G)$.
\end{theorem}

\begin{proof}
We prove the theorem for $\R$. The statements for
$\D$, $\R_0$, and $\Q$ will follow analogously.
Let $\phi:\B^{N_2}\to\B^G$ be the best possible extension
function, so that
$\R^G(\Mar(f))=\max_{g\in\Mar(f)}\R(g^{\phi})$.

We now describe a randomized communication
protocol for computing $f$. Alice receives a string
$a$, and must compute $f(a,b)$, where $b$ is Bob's string.
This is equivalent to computing $f_a(b)$.
Since Alice knows $f_a$, she also knows $f_a^{\phi}$.
Let $R$ a randomized algorithm that computes
$f_a^{\phi}$ using at most $\R^G(\Mar(f))$ queries.
Alice will run this algorithm, and for each query,
she will send the index of that query to Bob
(as a number between $1$ and $G$). Bob will reply
with the corresponding bit of $\phi(y)$ (as
a bit in \{0,1\}). This allows Alice
to compute $f_a(b)=f(a,b)$.

The total communication in this protocol is at most
$(\lceil\log G\rceil+1)\R^G(\Mar(f))$.
Since this upper-bounds the randomized communication
complexity of $f$ (using private coins), the desired
result follows.
\end{proof}

\thm{communication_extension} allows us to use communication
complexity as a tool for lower-bounding the extended
query complexity of certain sets of functions. To use
it to solve the extended queries problem,
we need a function $f$ that has large randomized
communication complexity but for which $\Q(\Mar(f))$
is small. To construct such a function,
we start from a simple function that was recently shown
to separate randomized from quantum communication complexity,
called the Vector in Subspace problem.

\textbf{The Vector in Subspace Problem.}
In this problem, Bob gets a
unit vector $v\in\mathbb{R}^n$, and Alice
gets a subspace $H$ of $\mathbb{R}^n$ of dimension $n/2$.
It is promised that either $v\in H$ or $v\in H^\perp$;
the task is to determine which is the case.
We assume for simplicity that $n$ is a power of $2$.

This problem was first studied in \cite{Kre95}
and was also described in \cite{Raz99}. Klartag
and Regev \cite{KR11} showed
that this problem has randomized
communication complexity $\Omega(n^{1/3})$.
In addition, it is easy to see that the one-way
quantum communication complexity of the problem is at most
$\log n$: Bob can send
a superposition over $\log n$ bits with amplitudes
determined by $v$; Alice can then apply
the projective measurement given by $(H,H^\perp)$.

To apply this function to the extended queries problem,
we need a few modifications. First, we need a discrete
version of the problem.
\cite{KR11} showed that a lower bound of $\Omega(n^{1/3})$
for randomized communication complexity applies
to a discrete version of the problem in which
each real number is described using $O(\log n)$ bits;
that is,
Alice's subspace is given using $n/2$ vectors
of length $n$, whose entries are specified using $O(\log n)$
bits, and Bob's vector is specified using $n$
real numbers of $O(\log n)$ bits each.

$\Mar(f)$ is the set of functions where we
know Alice's subspace $H$, and are allowed to query from
Bob's input vector. However, phrased this way, it is
not clear how to use a quantum algorithm to compute
such functions using few queries. To solve this problem,
we modify the way Bob's input is specified.
Instead of specifying only the entries to the vector,
Bob's input string also lists some ``partial sums'' of
the vector entries.

The idea is for
Bob's vector to allow Alice to use the following
algorithm to construct the state with amplitudes specified
by $v$. We interpret $v$ as specifying a superposition
over strings of length $\log n$.
Alice starts by querying the probability $p$ that the first
bit of this string is $0$ when this state
is measured. Alice will now place a $\sqrt{p}$ amplitude
on querying the probability that the second bit is $0$
conditioned on the first bit being $0$, and a $1-\sqrt{p}$
amplitude on querying the probability that the second bit
is $0$ conditioned on the first bit being $1$.
Alice keeps going in this way, until she gets to the final
bit of the string of length $\log n$, at which point
she queries the phase. This allows her to construct
the state determined by the amplitudes in $v$.

Of course, for this to work, Bob's input must
provide all of these conditional probabilities.
There is one such probability to specify for the first
bit, two for the second, four for the third, and so on.
Since there are $\log n$ bits, Bob's input needs to
specify only $O(n)$ probabilities. Each can be specified
with $O(\log n)$ precision, so Bob's total input
size is $O(n\log n)$. Moreover, Alice constructs
the desired state after $O(\log n)$ queries to the
probabilities, or $O(\log^2 n)$ queries to the bits
of Bob's input.

We thus get the following theorem.

\begin{theorem}\label{thm:Qextension}
For all $A\in\mathbb{N}$, there is a set $S$ of partial functions
from $\B^A$ to $\B$ such that for all $G\geq A$,
\[\Q(S)=O(\log^2 A),\quad
\R^G(S)=\Omega\left(\frac{A^{1/3}}{\log^{1/3}A\cdot\log G}\right).\]
\end{theorem}

\begin{proof}
Let $f$ be the function described above
with $n=A/\log A$, and let $S=\Mar(f)$.
Then $\Q(S)=O(\log^2 n)=O(\log^2 A)$ and
$\R^{\CC}(f)=\Omega(n^{1/3})=\Omega(A^{1/3}/\log^{1/3}A)$.
By \thm{communication_extension}, we get
$\R^G(S)=\Omega(A^{1/3}/(\log^{1/3}A\cdot\log G))$.
\end{proof}

Together with \thm{sculpting_extension},
this implies that any function with large $\Hi(\RC_f)$
can be sculpted, simply by plugging
$S$ from \thm{Qextension} into \thm{sculpting_extension}
and setting $G=N$.

\subsection{Reducing Sculpting to Extended Query Complexity}\label{sec:sculpting_extension}

We now prove \thm{sculpting_extension},
restated here for convenience.

\sculptingExtension*

\begin{proof}
There are at least $2^{\Hi(\RC_f\cdot 2\log N)}$
inputs $x$ with
$\RC_f(x)\geq\Hi(\RC_f\cdot 2\log N)/(2\log N)$.
Let the set of such inputs be $C$.
By \lem{sauer}, if $N\geq 2$, there is
a set $B$ of
\[\frac{\Hi(\RC_f\cdot 2\log N)}{2\log N}
\geq\frac{\Hi(\RC_f)}{2\log N}\]
indices in $\{1,2,\dots,N\}$ that is shattered
by the inputs in $C$. We'll restrict $B$ to have
size at most $\Hi(\RC_f)/(4\log N)$,
so $|B|=A$.
Let $\phi:\B^A\to\B^N$
be defined by mapping each string $x\in\B^A$ to
a string $z$ in $C$ such that restricting $z$
to $A$ gives $x$.
This is an injective mapping, so $\phi$ is an extension
function with extension size $N$.

Next, consider the set $S$ of partial Boolean functions
from $\B^A$ to $\B$.
Let $S^{\phi}=\{g^{\phi}:g\in S\}$.
Then $\R(S^\phi)\geq \R^N(S)$. It follows that there
is some function $g^\phi\in S^\phi$ such that
$\R(g^\phi)\geq\R^N(S)$.

We will use the function $g^\phi$ to define the desired
promise $P$. The domain of $g^\phi$ is contained in $C$.
Let $x$ be in this domain, so that
$\RC_f(x)\geq\Hi(\RC_f\cdot 2\log N)/(2\log N)\geq 2A$.
Let $\mu_x$ be a distribution over inputs $y$
such that $f(x)\neq f(y)$, with the property that
for any bit $i$,
$\Pr[y_i\neq x_i]\leq1/\RC_f(x)\leq 1/(2A)$.
Then for all $x\in C$, a randomized algorithm has a
hard time distinguishing between $x$ and $\mu_x$.
For each such $x$, let $\mu^\prime_x$ be the distribution
$\mu_x$ conditioned on the sampled input
agreeing with $x$ on the bits in $B$.
Since the probability of an input sampled from $\mu_x$
disagreeing with $x$ on $B$ is at most
$|B|\cdot 1/(2A)\leq 1/2$, the distribution $\mu_x^\prime$
is not too far from $\mu_x$. In particular,
any randomized algorithm that finds a disagreement with
$x$ on an input sampled from $\mu_x^\prime$
with probability $p$ will also find a disagreement with
$x$ on an input sampled from $\mu_x$ with probability
at least $p/2$. It follows that a randomized
algorithm must use $\Omega(A)$ queries to distinguish
$x$ from $\mu_x^\prime$.

We construct the promise $P$ as follows.
Start with $P=\emptyset$. For each
$x\in\Dom(g^\phi)$, we add $x$ to $P$
if $f(x)=g^\phi(x)$; otherwise, we add the
support of $\mu_x^\prime$ to $P$.

It remains to lower-bound $\R(f|_P)$ and
to upper-bound $\Q(f|_P)$. We start with the
upper bound. Let $y\in P$, and consider the value
of $y$ on $B$. If $x$ is an input of the domain
of $g^\phi$ that caused $y$ to be added, then
$x$ and $y$ agree on $B$. Further, the values
of $x$ on $B$ are simply
$\phi^{-1}(x)\in\B^{|B|}$, and
$g(\phi^{-1}(x))=g^\phi(x)=f(y)$.
This means $g(y|_B)=f(y)$.
We now have the quantum algorithm work only with the bits
of $y|_B$, ignoring the rest. The algorithm need only
compute $g(y|_B)$.
Since $g\in S$, we get
$\Q(f|_P)\leq\Q(g)\leq\Q(S)$, as desired.
A similar argument would upper-bound
other complexity measures, such as $\R$, $\R_0$, or $\D$.

For the lower bound, consider the hard distribution
$\mu$ on inputs to $g^\phi$ obtained
from Yao's minimax principle \cite{Yao77}.
This distribution has the property that any randomized
algorithm for $g^\phi$ that succeeds with probability at
least $2/3$ on inputs sampled from $\mu$ must
use $\R(g^\phi)$ queries.
We construct a new
distribution $\mu^\prime$ over $P$ by generating
an element $x\in\Dom(g^\phi)$ according to $\mu$,
and then outputting either $x$ or a sample from
$\mu_x^\prime$, depending on which of them was added to $P$.
We lower-bound the number of queries a randomized
algorithm requires to compute $f$ on an
input sampled from $\mu^\prime$ by a reduction
from either computing $g^\phi$
on inputs sampled from $\mu$,
or else distinguishing $x$ from $\mu^\prime_x$.

Let $R$ be a randomized algorithm for $f|_P$.
Let $x\sim\mu$. We wish to compute $g^\phi(x)$.
Although $x$ may not be in $P$,
consider running $R$ on $x$ anyway. The algorithm
will correctly output $g^\phi(x)$ with some probability $p$,
depending on both the internal randomness of $R$ and
on $\mu$. If $p\geq 3/5$, we could
amplify $R$ a constant number of times to turn it
into an algorithm for $g$ that works on inputs
sampled from the hard distribution $\mu$,
which means $R$ must
use $\Omega(\R(g^\phi))=\Omega(\R^N(S))$ queries.
So suppose that $p\leq 3/5$.

Next, given $x\sim\mu$, we let $y_x$ be either $x$
or a sample from $\mu^\prime_x$,
as $\mu^\prime$ dictates.
Then running $R$ on $y_x$ gives $f(y_x)=g^\phi(x)$
with probability at least $2/3$. On the other hand,
running $R$ on $x$ gives output $g^\phi(x)$ with
probability at most $3/5$. That is, we have
\[\Pr_{R,x\sim\mu}[R(x)=g^\phi(x)]
    =\E_{x\sim\mu}\left[\Pr_R[R(x)=g^\phi(x)]\right]
    \leq 3/5\]
\[\Pr_{R,x\sim\mu,y_x}[R(y_x)=g^\phi(x)]
    =\E_{x\sim\mu}\left[\Pr_{R,y_x}[R(y_x)=g^\phi(x)]\right]
    \geq 2/3\]
From which it follows that
\[\E_{x\sim\mu}
    \left[\Pr_{R,y_x}[R(y_x)=g^\phi(x)]-
        \Pr_R[R(x)=g^\phi(x)]\right]
    \geq 1/15.\]
This means there
must be some specific input $\hat{x}$ such that
the probability of $R$ outputting $g^\phi(\hat{x})$ when
run on $y_{\hat{x}}$ is at least $1/15$ more than the
probability of $R$ outputting $g^\phi(\hat{x})$ when run
on $\hat{x}$.
In particular, we must have $y_{\hat{x}}\neq \hat{x}$,
so $y_{\hat{x}}$ is a sample from $\mu^\prime_{\hat{x}}$.
Therefore, $R$ distinguishes $\hat{x}$ from
$\mu^\prime_{\hat{x}}$ with constant probability,
so it uses at least $\Omega(A)$ queries.

We conclude that
$\R(f|_P)=\Omega(\min\{\R^N(S),A\})$.
Since the domain of the functions in $S$ is $\B^A$,
their query complexity is at most $A$.
Thus $\R(f|_P)=\Omega(\R^N(S))$, as desired.
A similar argument lower-bounds other complexity
measures, such as $\R_0$ or $\D$.
\end{proof}

This proof uses
the fact that $\RC$ lower-bounds $\R$,
so it would not work on complexity measures that
are not lower-bounded by $\RC$ (for example, $\C^{(1)}$).
For $\Q$, it might be possible to use a similar argument
and suffer a quadratic loss, since $\Q$ is lower-bounded
by $\sqrt{\RC}$. However, since there is no hard
distribution for a quantum query complexity problem,
this might be trickier to prove (we will not need it
in this paper).

We can use the previous theorems to get a sculptability
result for $\R$ vs.\ $\Q$ in terms of the H-index of
randomized certificate complexity.

\begin{corollary}\label{cor:sculpting_upper}
Let $f:\B^N\to\B$ be a total function. Then
there is a promise $P\subseteq\B^N$ such that
\[\Q(f|_P)=O(\log^2\Hi(\RC_f)),\qquad
    \R(f|_P)=\Omega\left(\frac{\Hi(\RC_f)^{1/3}}
            {\log^{5/3} N}\right).\]
There is also a promise $P^\prime\subseteq\B^N$
such that
\[\R_0(f|_{P^\prime})=O(1),\qquad
    \D(f|_{P^\prime})=\Omega\left(\frac{\Hi(\RC_f)}
            {\log^2 N}\right).\]
\end{corollary}

\begin{proof}
This follows from \thm{sculpting_extension} together with
\thm{Qextension} and \thm{R0extension}.
\end{proof}

To complete the proof of \thm{main}, all that remains
is relating $\Hi(\RC_f)$ to $\Hi(\C_f)$.

\section{Relating $\Hi(\C_f)$, $\Hi(\RC_f)$, and $\Hi(\bs_f)$}
\label{sec:Hprops}

In this section, we relate $\Hi(\C_f)$ to $\Hi(\RC_f)$,
completing the characterization of sculpting.
Actually, we will prove a relationship between
$\Hi(\C_f)$ and $\Hi(\bs_f)$, which implies the
desired relationship since $\Hi(\bs_f)\leq\Hi(\RC_f)$.
The proof is somewhat involved, but splits naturally
into three parts. In \lem{CvsRC}, we show a relationship
between $\C_f(x)$ and $\RC_f(x)$ in terms of the number
of $0$- and $1$-inputs of $f$.
In \thm{Hbal}, we show that $\Hi(\C_f)=O(\Bal(f)\log N)$.
Finally, \thm{Hbs} gives the desired relationship
between $\Hi(\C_f)$ and $\Hi(\bs_f)$.

\begin{lemma}\label{lem:CvsRC}
    Let $f:\B^N\to\{0,1\}$ be a partial function,
    and let $x\in\Dom(f)$. If $f(x)=0$, then
    \[\C_f(x)\leq\RC_f(x)(1+\log|f^{-1}(1)|)\]
    and if $f(x)=1$, then
    \[\C_f(x)\leq\RC_f(x)(1+\log|f^{-1}(0)|).\]
\end{lemma}

\begin{proof}
For $x\in f^{-1}(1)$, we wish to upper-bound $\C_f(x)$
in terms of $\RC_f(x)$, assuming $f^{-1}(0)$ is small.
A certificate for $x$ consists of a partial assignment
of $x$ that contradicts
all the elements of $f^{-1}(0)$.

Consider the greedy strategy for certifying $x$,
which works by repeatedly choosing the bit of $x$
that contradicts as many of the $0$-inputs as possible,
and adding it to the certificate.
By definition, this strategy produces a certificate for
$x$ of size at least $\C_f(x)$.

Let $p_i$ be the fraction of the remaining inputs
which are contradicted by the $i$-th bit of the greedy
algorithm. The number of remaining inputs
during the run of the greedy algorithm is then
\[|f^{-1}(0)|,\;|f^{-1}(0)|(1-p_1),\;
|f^{-1}(0)|(1-p_1)(1-p_2),\dots\]

The number of remaining inputs in the greedy
algorithm will be upper-bounded by a geometric sequence
that starts at $|f^{-1}(0)|$ and has ratio
$1-\min_i p_i$.
Such a sequence decreases to $1$ after at most
\[\frac{-1}{\log(1-\min_i p_i))}(1+\log|f^{-1}(0)|)
    \leq\frac{1+\log|f^{-1}(0)|}{\min_i p_i}\]
steps. It follows that
\[\C_f(x)\leq\frac{1+\log|f^{-1}(0)|}{\min_i p_i}.\]

It remains to show that $\RC_f(x)=\Omega(1/\min_i p_i)$.
Let $j$ be the step of the greedy algorithm that
achieves this minimum, i.e. $p_j=\min_i p_i$.
Then before the $j$\textsuperscript{th}
step of the algorithm,
there is a non-empty set $S$ of $0$-inputs for $f$
such that for any bit of the input, at most a $p_j$
fraction of the elements of $S$
disagree with $x$ on that bit. In other words,
$x$ is entry-wise very close to the ``average''
of the elements of $S$. If we give each
element of $S$ weight $1/(p_j|S|)$, we would
get a feasible set of fractional blocks with
total weight $1/p_j$. Thus $\RC_f(x)\geq 1/p_j$,
so $\C_f(x)\leq \RC_f(x)(\log|f^{-1}(0)|+1)$.
An analogous argument works when $x$ is a $0$-input to $f$.
\end{proof}

\begin{theorem} \label{thm:Hbal}
    Let $N\geq 2$, and let $f:\B^N\to\{0,1\}$
    be a total function. Then
    \[\Hi(\C_f)\leq 10\Bal(f)\log N.\]
\end{theorem}

\noindent
\begin{proof}
Without loss of generality, suppose
$|f^{-1}(0)|\leq|f^{-1}(1)|$.
The number of $0$-inputs with large certificates is
at most $|f^{-1}(0)|\leq 2^{\Bal(f)}$.
Let $S$ be the set of $1$-inputs with certificates
of size greater than $5\Bal(f)$.
We wish to show that $S$ is small.
\lem{sauer} implies there is a set
$B=\{i_1,i_2,\ldots,i_{|B|}\}$ of indices
of the input of size at least
$\log|S|/(2\log N)$ that is shattered by $S$.
Therefore, to show that $S$ is small, it suffices
to show that $B$ is small.

From \lem{CvsRC}, we have
$\C_f(x)\leq\RC_f(x)\Bal(f)$
for any $1$-input $x$, so for all $x\in S$, we have
$\RC_f(x)\geq \C_f(x)/\Bal(f)>5$.
This means for all $x\in S$, there is
a distribution $\mu_x$ over $0$-inputs such that
for each $i$, the probability that $y_i\neq x_i$
when $y$ is sampled from $\mu_x$ is less than $1/5$.

Let $\mu_B$ be the uniform distribution over $B$.
Let $\delta(b,c)=1$ if $b\neq c$ and $0$ otherwise.
We then write
\[\frac{1}{5}
    >\mathop{\mathbb{E}}_{i\sim\mu_B}\left(
\mathop{\mathbb{E}}_{y\sim\mu_x}\delta(x_i,y_i)\right)
=\mathop{\mathbb{E}}_{y\sim\mu_x}\left(
\mathop{\mathbb{E}}_{i\sim\mu_B}\delta(x_i,y_i)\right).\]
We can conclude that for any $x\in S$, there exists
a $0$-input $y_x$ that differs from $x$
in less than one fifth of the bits of $B$.
In other words, the distance between $x|_B$ and $y_x|_B$ is
less than $|B|/5$. The idea is now
to upper-bound $|B|$ by using the fact
that for every string in $\B^{|B|}$ there is a $0$ input $y$
such that $y|_B$ is close to that string,
and there are few $0$-inputs overall.
Indeed, the number of strings in
$\B^{|B|}$ is $2^{|B|}$, and each $0$-input
can only be of distance less than $|B|/5$
from $2^{H(1/5)|B|}$ of them (where $H(1/5)$
denotes the entropy of $1/5$). Therefore,
to cover all the strings in $\B^{|B|}$,
there must be more than $2^{(1-H(1/5))|B|}$
$0$-inputs. Then
\[\Bal(f)\geq\log |f^{-1}(0)|> (1-H(1/5))|B|
\geq (1-H(1/5))\frac{\log |S|}{2\log N},\]
so
\[\log |S|< \frac{2\Bal(f)\log N}{1-H(1/5)}
\leq 8\Bal(f)\log N.\]
This means there are less than
$2^{8\Bal(f)\log N}$
$1$-inputs with certificate size at least $5\Bal(f)$.
There are also at most $2^{\Bal(f)}$ $0$-inputs
with certificate size at least $5\Bal(f)$
(because there are at most that many $0$-inputs in total).
Thus the log of the total number of inputs with certificates
larger than $5\Bal(f)$ is at most
$10\Bal(f)\log N$.
It follows that $\Hi(\C_f)\leq 10\Bal(f)\log N$.
\end{proof}

\begin{theorem} \label{thm:Hbs}
    Let $f:\B^N\to\B$ be a total function. Then
    \[\Hi(\C_f)=O(\Hi(\bs_f)^2\log N).\]
\end{theorem}

\begin{proof}
Let $A$ be the set of inputs that have certificate size
more than $\Hi(\C_f)$.
Let $A_0$ be the set of $0$-inputs in $A$,
and let $A_1$ be the set of $1$-inputs in $A$.
Let $B$ be the set of inputs that have block sensitivity
more than $b$, with $b=\sqrt{\Hi(\C_f)/2}$.
Let $B_0$ be the set of $0$-inputs in $B$,
and let $B_1$ be the set of $1$-inputs in $B$.
Without loss of generality, assume $|A_0|\geq |A_1|$.
Since $|A|\geq 2^{\Hi(\C_f)}$, we have $|A_0|\geq 2^{\Hi(\C_f)-1}$.

Now, let $g:\B^N\to\{0,1\}$ be the total function
defined by $g(x)=1$ if and only if $x\in B_1$.
Suppose $x$ is an element of $A_0\backslash B_0$.
Consider certifying that $x$ is a $0$-input for $g$;
let $p$ be the smallest such certificate.
Then $p$ is consistent with $x$ but inconsistent
with all the strings in $B_1$. We claim that this
certificate must be large: its size must be greater
than $\Hi(\C_f)-b^2=\Hi(\C_f)/2$. To show this, we show
that we can turn $p$ into a certificate for $x$ with
respect to $f$ (instead of with respect to $g$)
by adding only $b^2$ bits to it.

Let $q$ be a minimal sensitive block of $x$
(with respect to $f$)
that is disjoint from $p$. Since $x$
is a $0$-input for $f$, $x^q$ is a $1$-input
for $f$. Since $q$ is disjoint from $p$,
$x^q$ is consistent with $p$, so $x^q\notin B_1$.
Thus the block sensitivity of $x^q$ is at most $b$.
However, since $q$ was a minimal sensitive block,
the sensitivity of $x^q$ is at least $|q|$;
thus $|q|\leq b$. It follows that
all minimal sensitive blocks of $x$
that are disjoint from $p$ must have size at most $b$.

In addition, since $x\in A_0\backslash B_0$, the block
sensitivity of $x$ is at most $b$.
We can now construct a certificate for $x$
by taking a maximal set of minimal disjoint sensitive
blocks for $x$, all of which are disjoint from $p$.
There will be at most $b$ such blocks, and each
will have size at most $b$.
Therefore, this certificate for $x$ has size at most
$|p|+b^2$. Since $x\in A_0$, we must have
$|p|+b^2>\Hi(\C_f)$, or $|p|>\Hi(\C_f)-b^2=\Hi(\C_f)/2$.
We have shown that the elements of $A_0\backslash B_0$
all have certificate size greater than $\Hi(\C_f)/2$
even with respect to $g$.

Now, by \thm{Hbal},
the number of inputs $x$ that have certificate
size more than $10(1+\log|B_1|)\log N$ with respect to $g$
is at most $2^{10(1+\log|B_1|)\log N}$.
It follows that either
$\Hi(\C_f)/2\leq 10(1+\log|B_1|)\log N$
(so that the theorem doesn't apply), or else
$|A_0\backslash B_0|\leq 2^{10(1+\log|B_1|)\log N}$.

In the former case, we have
\[\log |B|\geq\log |B_1|
\geq\frac{\Hi(\C_f)}{20\log N}-1.\]
In the latter case, we have 
\[2^{\Hi(\C_f)-1}\leq |A_0|\leq |B_0|+
2^{10(1+\log|B_1|)\log N}
=|B_0|+(2|B_1|)^{10\log N}
\leq (2|B|)^{10\log N},\]
so in that case,
\[\log |B| \geq
    \frac{\Hi(\C_f)-1}{10\log N}-1.\]
Thus, in both cases,
\[\log |B|\geq\frac{\Hi(\C_f)-1}{20\log N}-1
=\Omega\left(\frac{\Hi(\C_f)}{\log N}\right).\]

This means there are $2^{\Omega(\Hi(\C_f)/\log N)}$
inputs with block sensitivity more than
$\sqrt{\Hi(\C_f)/2}$. We thus have
\[\Hi(\bs_f)=\Omega\left(\min\left\{
\frac{\Hi(\C_f)}{\log N},\sqrt{\Hi(\C_f)}\right\}\right)
=\Omega\left(\sqrt{\frac{\Hi(\C_f)}{\log N}}\right).\]
\end{proof}

\thm{main} now follows from \thm{RvsQ}
(the non-sculptability theorem in terms of $\Hi(\C_f^2)$),
\cor{sculpting_upper}
(the sculptability result in terms of $\Hi(\RC_f)$),
and \thm{Hbs}
(relating $\Hi(\bs_f)$ to $\Hi(\C_f)$), together
with the properties that $\Hi(\C_f^2)\leq\Hi(\C_f)^2$
and that $\Hi(\bs_f)\leq\Hi(\RC_f)$. We restate
\thm{main} here for convenience.

\main*

\thm{main0} follows as a corollary. This completes
the proof of our main result.

\section{Sculpting Randomized Speedups}

Now that we've characterized sculpting quantum
query complexity, we turn our attention to sculpting
other measures. Recall that
\[\Q(f)\leq\R(f)\leq\R_0(f)\leq\D(f).\]
We showed that sculpting $\R(f)$ vs.\ $\Q(f)$
is possible if and
only if $f$ has a large number of large certificates.
We now show that the exact same condition characterizes
sculpting $\D(f)$ vs.\ $\R_0(f)$.
On the other hand, we show that $\R_0(f)$ vs.\ $\R(f)$
behaves differently: it's \emph{always} possible
to sculpt a function $f$ to a promise $P$ such that
$\R(f|_P)$ is constant and $\R_0(f|_P)$ is almost
as large as $\R_0(f)$.

We start by characterizing sculpting for $\D$ vs.\ $\R_0$.

\subsection{Sculpting $\D$ vs.\ $\R_0$}

The proof of this characterization will follow that
of \thm{main}. For the non-sculptability direction,
we need an analogue of \thm{RvsQ}, relating
deterministic and zero-error randomized query complexities
in terms of $\Bal(f)$. We prove the following theorem.

\begin{theorem}\label{thm:DvsR0}
    Let $f:\B^N\to\{0,1\}$ be a partial function. Then
    \[D(f)\leq 2R_0(f)\Bal(f).\]
\end{theorem}

\begin{proof}
Consider the zero-error
randomized algorithm that takes $\R_0(f)$
expected queries to evaluate $f$.
By Markov's inequality, if we let this algorithm
make $2\R_0(f)$ queries on input $x$, it will succeed
in computing $f(x)$ (and provide a certificate for $x$)
with probability at least $1/2$.
This gives us a probability distribution $\mu$ over
deterministic algorithms, each of which makes
$2\R_0(f)$ queries, such that for each input $x$
the probability that an
algorithm sampled from $\mu$ finds a certificate when run on
$x$ is at least $1/2$.

For a deterministic algorithm $D$ and an input $x$,
let $c(D,x)=1$ if $D$ finds a certificate for $x$,
and $c(D,x)=0$ otherwise. Let $Z\subseteq\B^N$. Then
\[\E_{D\sim\mu}\left[\sum_{x\in Z}c(D,x)\right]
=\sum_{x\in Z}\E_{D\sim\mu}[c(D,x)]
\geq\sum_{x\in Z}(1/2)
=\frac{|Z|}{2}.\]
It follows that there is a deterministic
algorithm $D_Z$ that makes $2\R_0(f)$ queries
and finds a certificate when run on half the inputs in $Z$.

Suppose without loss of generality
that $|f^{-1}(0)|\leq |f^{-1}(1)|$.
Now, on input $x$, set $Z=f^{-1}(0)$,
and run $D_Z$ on $x$.
If it fails to find a certificate,
then we have eliminated half of $Z$
as possibilities for the input. Repeating this
$\lfloor\log |f^{-1}(0)|\rfloor+1\leq\Bal(f)$
times suffices to eliminate all of $f^{-1}(0)$
as possibilities for $x$,
and hence to determine the value of $f(x)$.
The total number of queries used is at most
$2\R_0(f)\Bal(f)$.
\end{proof}

Note that \thm{DvsR0} and \thm{RvsQ}
together complete the proof of \thm{bal}.

Next, we turn \thm{DvsR0} into a non-sculptability
theorem in terms of $\Hi(\C_f)$. The argument in
\thm{noSculpting} follows verbatim,
and we get the following
sculpting lower bound.

\begin{corollary}\label{cor:R0H}
Let $f:\B^N\to\B$ be a total function. For any promise
$P\subseteq\B^N$, we have
\[\D(f|_P)=O(\R_0(f|_P)\Hi(\C_f)^2).\]
\end{corollary}

We now prove the other direction: we show
that sculpting is possible when $\Hi(\RC_f)$
is large. Using the arguments from \sec{sculpt},
it suffices to solve the extended queries problem
for $\D$ vs.\ $\R_0$. We do this using the reduction
to communication complexity in
\thm{communication_extension}.

\begin{theorem}\label{thm:R0extension}
For all $N\in\mathbb{N}$, there is a set of partial
functions $S$ from $\B^N$ to $\B$ such that
for all $G\geq N$,
\[\R_0(S)=O(1),\qquad
\D^G(S)=\Omega\left(\frac{N}{\log G}\right).\]
\end{theorem}

\begin{proof}
We start with \textsc{Equality},
in which Alice and Bob are each given an $n$-bit
string and wish to know if their strings are equal.
This problem has deterministic query complexity
$\Omega(n)$, but small randomized query complexity.
To make the zero-error randomized query complexity
small as well, we give Alice and Bob two strings each,
with the promise that either their first strings
are equal and the second strings are not, or vice versa.
The goal will be to determine which is the case.
It is not hard to see that the deterministic communication
complexity of this problem is still $\Omega(n)$.

We need to get the zero-error
randomized query complexity of the marginal
functions to be small. To do this, we introduce another
modification: we encode each of Bob's strings using a fixed
random code of length $3n$. This code will have the property
that the distance between any pair of codewords is
$\Omega(n)$. To compute a marginal function $f_{a_1,a_2}$
indexed by Alice's strings,
we can simply randomly sample from each of Bob's strings;
after $O(1)$ samples, we will discover which of
his strings do not match the codeword corresponding
to $a_1$ and $a_2$.

This construction gives us a function
$f:\B^{2n}\times\B^{6n}\to\B$
such that $\D^{\CC}(f)=\Omega(n)$ and
$\R_0(\Mar(f))=O(1)$.
Setting $N=6n$ and using \thm{communication_extension}
finishes the proof.
\end{proof}

Putting this together, we get the following
sculpting theorem which, together with
\cor{R0H}, is analogous to \thm{main}.

\begin{theorem}\label{thm:R0main}
For all total functions $f:\B^N\to\B$,
there is a promise $P\subseteq\B^N$ such that
\[\D(f|_P)=\Omega\left(
\frac{\sqrt{\Hi(\C_f)}}{\log^{5/2}N}\right)
\qquad\mbox{and}\qquad\R_0(f|_P)=O(1).\]
\end{theorem}

\begin{proof}
This follows from \thm{R0extension} together with
\thm{sculpting_extension} and \thm{Hbs}.
\end{proof}

We also get the following corollary, analogous to
\thm{main0}.

\begin{corollary}\label{R0basic}
Let $f:\B^N\to\B$ be a total function. Then there is
a promise $P\subseteq\B^N$ such that
$\D(f|_P)=N^{\Omega(1)}$ and $\R_0(f|_P)=N^{o(1)}$,
if and only if $\Hi(\C_f)=N^{\Omega(1)}$.
Futhermore, in this case we also have
$\R_0(f|_P)=O(1)$.
\end{corollary}

\subsection{Sculpting $\R_0$ vs.\ $\R$}

While it is possible to use the above argument
to get a sculptability result for $\R_0$ vs.\ $\R$,
we can get a stronger result by a direct argument.
In fact, unlike $\R$ vs.\ $\Q$ or $\D$ vs.\ $\R_0$,
sculpting $\R_0$ vs.\ $\R$ is \emph{always} possible
(there is no dependence on any H-index).

\begin{theorem}\label{Rsculpting}
Let $f:\B^N\to\B$ be a non-constant total function.
Then there is a promise $P\subseteq\B^N$
such that
\[\R(f|_P)=1,\qquad\R_0(f|_P)\geq\frac{\R_0(f)^{1/3}}{6}.\]
\end{theorem}

\begin{proof}
We actually prove a stronger result,
finding a promise $P$ such that
$\R(f|_P)=1$ and $\R_0(f|_P)\geq\bs(f)/6$.
We then use the known relationship $\R_0(f)\leq\bs(f)^3$
for total functions to get the desired result.
Note that finding $P$ with 
$\R(f|_P)=1$ and $\R_0(f|_P)\geq \bs(f)/6$
is trivial when $\bs(f)\leq 6$; thus we assume $\bs(f)>6$.

Let $x\in\B^N$ be such that $\bs_f(x)=\bs(f)$.
Assume without loss of generality that $f(x)=0$.
Let $S_1$ be the set of all $1$-inputs with Hamming
distance at least $(2/3)N$ from $x$.
For any partial assignment $p$ consistent with $x$,
let $S_1^p\subseteq S_1$
be the set of all inputs $y$ in $S_1$
that are consistent with $p$.

There are two cases. If $S_1^p$ is non-empty for
all partial assignments $p$ consistent with $x$
of size less than $\bs_f(x)/6$, then we can pick
the promise to be $P=\{x\}\cup S_1$.
It then follows that certifying that $f|_P$ is $0$
on input $x$ takes at least $\bs_f(x)/6$ queries,
whence $\R_0(f|_P)\geq \bs_f(x)/6$.
On the other hand, a randomized algorithm
can make $1$ query to check if the input differs from
$x$. Thus $\R(f|_P)=1$.

The other case is that there is some partial assignment
$p$ of size less than $\bs_f(x)/6$ such that $S_1^p$
is empty. We restrict our attention to inputs
consistent with $p$. Since $x$ has $\bs_f(x)$
disjoint sensitive blocks,
it has at least $(5/6)\bs_f(x)$ disjoint sensitive
blocks that do not overlap with $p$.
We exclude blocks of size larger than $N/3$. Since there
are at most $2$ such blocks, this leaves at
least $(5/6)\bs_f(x)-2$.
Let $B$ be the set of inputs we get by flipping one
of these blocks of $x$. Then $B$ contains only $1$-inputs
to $f$ that are consistent with $p$, all of which
have Hamming distance at most $N/3$ from $x$.
Since $\bs_f(x)=\bs(f)>6$, we have $B\neq\emptyset$.

Let $S$ be the set of inputs consistent with $p$
that have Hamming distance at least $(2/3)N$ from $x$.
Since $S_1^p$ is empty, $S$ contains only $0$-inputs to $f$.
Let $P=B\cup S$. Now, consider certifying that an
input $y$ to $f|_P$ is a $1$-input. Since all inputs
of Hamming distance at least $(2/3)N$ from $x$ that
are consistent with $p$ are $0$-inputs, this requires
showing at least $N/3-|p|$ bits of $y$. Since
$|p|<\bs_f(x)/6\leq N/6$, this is at least $N/6$.
Thus $\R_0(f|_P)\geq N/6\geq\bs(f)/6$.

On the other hand, a
bounded-error randomized algorithm can simply
query a bit of the input at random, and check for agreement
with $x$. If the bit agrees, the algorithm can output $1$,
and if the bit disagrees, the algorithm can output $0$.
This works because $0$-inputs have distance at least
$(2/3)N$ from $x$, while all $1$-inputs have distance at
most $N/3$ from $x$ (since the sensitive blocks used
to construct $B$ were of size at most $N/3$).
Thus $\R(f|_P)=1$.
\end{proof}

\section{Other Query Complexity Results}
\label{sec:other}

We can use some of the tools introduced in the previous
sections to prove some new relations in query
complexity. In \sec{small_dom}, we prove a quadratic
relationship between $\D(f)$ and $\Q(f)$
for partial functions $f$ that have small domain.
In \sec{tot}, we prove a quadratic relationship between
$\D(f)$ and $\Q(f)$ for total functions $f$ for which
$\Hi(\C_f)$ is small.

\subsection{Query Complexity on Small
Promises}\label{sec:small_dom}

We prove \thm{small}, which we restate for convenience.

\small*

\begin{proof}
We follow the proof of \thm{RvsQ}. The randomized algorithm
used in that proof relies only on the existence
of a randomized algorithm distinguishing a string
$a\in\B^N$ from either $f^{-1}(0)$ or $f^{-1}(1)$,
which is in turn guaranteed by \lem{RvsQlemma}.
To make that algorithm deterministic, we only need to turn
this distinguishing algorithm into a deterministic one.
By \lem{CvsRC}, we have $\C_f(x)=O(\RC_f(x)\log|\Dom(f)|)$.
On the task of distinguishing a single input from
a set of inputs, certificate complexity equals deterministic
query complexity. Using this observation,
we can modify the proof of \thm{RvsQ} to get the
result
\[\D(f)=O(\Q(f)^2\Bal(f)\log|\Dom(f)|)
=O(\Q(f)^2\log^2|\Dom(f)|),\]
from which the desired result follows.
\end{proof}


\subsection{Relationship for Total Functions}
\label{sec:tot}

We can use H-indices to improve some of the relationships
between complexity measures on total functions,
proving \thm{tot}.
Recall that for total functions, we have
$\D(f)\leq \C(f)\bs(f)$ and $\bs(f)=O(\Q(f)^2)$,
from which we have $\D(f)=O(\Q(f)^2\C(f))$.
We strengthen this result to
$\D(f)=O(\Q(f)^2\Hi(\sqrt{\C_f})^2)$ for total
Boolean functions. Since
$\Hi(\sqrt{\C_f})\leq\sqrt{\C(f)}$,
this result is always stronger.
In addition, since $\C(\OR)=n$ and $\Hi(\C_{\OR})=1$,
this improvement is sometimes very strong.

We restate \thm{tot} for convenience.

\tot*

\begin{proof}
The proof follows the proof that $\D(f)\leq\C(f)\bs(f)$
\cite{BBC+01}.
We start by reviewing this proof.
The deterministic algorithm repeatedly picks possible
$0$-certificates that are consistent with the
input observed so far, and queries the entries
of these certificates. If the queried entries match
the $0$-certificate, the algorithm is done (the
value of $f(x)$ is known to be $0$).
If ever there are no additional $0$-certificates
consistent with the observed part of the input,
the value of the function is known to be $1$.

The key insight is that if this process repeats $k$
times, then the block sensitivity of the function
is at least $k$. Indeed, let $p$ be the partial
assignment revealed after $k$ iterations.
Pick a $1$-input $y$ for $f$ that is consistent with $p$.
Let $B_i$ be the set of entries queried in the $i$-th
iteration of the algorithm. Then
for each $i$, there is a way
to change only the variables in $B_i$ to form
a $0$-certificate for $f$. It follows that each
$B_i$ contains a sensitive block for $y$. Since the
$B_i$ sets are disjoint, we get $\bs_f(y)\geq k$,
so $\bs(f)\geq k$.

We modify the algorithm as follows. In each step,
we only allow the algorithm to pick $0$-certificates
that are of size at most $\Hi(\sqrt{\C_f})^2$.
Thus the algorithm uses at most
$\bs(f)\Hi(\sqrt{\C_f})^2$ queries
before it gets stuck. When it gets stuck, either
the value of $f$ on the input is determined,
or else there are no more $0$-certificates that are small
enough.

Next, we repeat the same process with $1$-certificates
instead of $0$-certificates. If the value of $f$
is not yet determined, it means that the input
is not consistent with any small enough certificates,
so the certificate
complexity of the input $x$ is greater than
$\Hi(\sqrt{\C_f})^2$.
This gives $\sqrt{\C_f(x)}>\Hi(\sqrt{\C_f})$.

By definition of the H-index, there are now at most
$2^{\Hi(\sqrt{\C_f})}$ possibilities for the input.
We've therefore restricted $f$ to a small domain $P$.
We now use \thm{small} to evaluate $f$
using
\[O(\Q(f)^2\log^2 |\Dom(f|_P)|)
=O(\Q(f)^2\Hi(\sqrt{\C_f})^2)\]
deterministic queries.
This is added to the
$\bs(f)\Hi(\sqrt{\C_f})^2$
queries from before. Using
$\bs(f)=O(\Q(f)^2)$, we get
\[\D(f)=O(\Q(f)^2\Hi(\sqrt{\C_f})^2).\]
\end{proof}


\section{Sculpting in the Computational Complexity Model}
\label{sec:TM}

In this section, we examine sculpting in the computational
complexity model. We start with some notation.
Given a language $L\subseteq\B^{\ast}$, we
let $L(x)\in\B$ be its characteristic function.
Also, given a language $L$ together with a promise
$P\subseteq\B^{\ast}$,
we let $L|_{P}$ be the promise problem of
distinguishing the set $P\cap L$ from the set $P\setminus L$.

Now, we call the language $L$ \emph{sculptable}
if there exists a promise
$P$, such that the promise problem $L|_{P}$ is in
$\mathsf{P{}romiseBQP}$ but not in $\mathsf{P{}romiseBPP}$.
We will use the following definition.

\begin{definition}[\cite{BH77}]\label{def:pad}
A language $L$ is called
\emph{paddable}, if there exists
a polynomial-time function $f(x,y)$ such that

\begin{enumerate}
\item[(1)] $f$ is polynomial-time invertible, and
\item[(2)] for all $x,y$, we have
    $x\in L\iff f(x,y) \in L$.
\end{enumerate}
\end{definition}

In other words, $L$ is paddable if it is possible to
\textquotedblleft pad out\textquotedblright
any input $x$ with irrelevant information $y$, in an
invertible way, without affecting membership in $L$.

The paddable languages were introduced by Berman and Hartmanis
\cite{BH77}, as part of their exploration of whether all
$\mathsf{NP}$-complete languages are
polynomial-time isomorphic: they showed
that the answer was `yes' for all \emph{paddable}
$\mathsf{NP}$-complete languages.
Under strong cryptographic assumptions, we now know that there
exist $\mathsf{NP}$-complete languages
that are neither paddable nor
isomorphic to each other \cite{KMR95}.
Nevertheless, it remains the case that
almost all the languages that
\textquotedblleft naturally arise in complexity
theory\textquotedblright\ are paddable.

Next, let us say that \emph{$\mathsf{P{}romiseBQP}$
is hard on average for
$\mathsf{P/poly}$} if there exists a promise problem
$H|_{S}\in\mathsf{P{}romiseBQP}$,
as well as a family of distributions
$\left\{  \mathcal{D}_n\right\}_n$
with support on the promise set $S$, such that

\begin{enumerate}
\item[(1)] $\mathcal{D}_n$ is samplable in classical
    $\operatorname*{poly}(n)$ time, and

\item[(2)] there is no family of classical circuits
    $\left\{C_{n}\right\}_n$,
    of size $\operatorname*{poly}(n)$,
    such that for all $n$,
\[\Pr_{y\sim\mathcal{D}_n}\left[ C_n(y)=H(y) \right] 
    \geq\frac{3}{4}.\]
\end{enumerate}

So for example, because of Shor's algorithm
\cite{Sho97}, combined with the
worst-case/average-case equivalence of the
discrete log problem, we can say
that \emph{if discrete log is not in
$\mathsf{P/poly}$\textit, then
$\mathsf{P{}romiseBQP}$ is hard on average for
$\mathsf{P/poly}$}.

We now prove \thm{pad}, which
we restate here for convenience.

\pad*

\begin{proof}
Let $L$ be a paddable language, and let
$f$\ be the padding function for $L$.
Also, let $H|_{S}$ be any problem
in $\mathsf{P{}romiseBQP}$ that is hard on average for
$\mathsf{P/poly}$, and
let $\left\{  \mathcal{D}_n\right\}_n$
be the associated family of hard
distributions.
Then we need to construct a promise,
$P\subseteq\B^{\ast}$,
such that the promise problem $L|_P$ is in
$\mathsf{P{}romiseBQP}$ but not in $\mathsf{P{}romiseBPP}$.

Our promise $P$ will simply consist of all inputs of the form
$f(x,y,a)$ such that $y\in S$ and
\[L(x) \equiv H(y) +a\left( \operatorname{mod} 2\right).\]
Here $a\in\B$ is a single bit, which we think of as
concatenated onto the end of $y$.

Clearly, $L|_P$ is in $\mathsf{P{}romiseBQP}$:
just invert $f$ to extract
the \textquotedblleft comment\textquotedblright
$(y,a)$, then compute
$H(y)  +a\left(\operatorname{mod}2\right)$.

We need to show that $L|_P$ is not in $\mathsf{P{}romiseBPP}$.
Suppose by contradiction that it was,
and let $\mathcal{A}$ be the algorithm such that
$\mathcal{A}(x) =L(x)$ for all $x\in P$. Then
we'll show how to either

\begin{enumerate}
\item[(1)] decide $L$ in $\mathsf{BPP}$ (with no promise), or

\item[(2)] decide $H$ in $\mathsf{P/poly}$,
    with high probability over $\mathcal{D}_{n}$.
\end{enumerate}

Given an arbitrary input $x\in\B^n$, imagine we do the
following: first sample $y\sim\mathcal{D}_n$, then run
$\mathcal{A}$ on the inputs $f(x,y,0)$ and
$f(x,y,1)$. There are two cases: first suppose
\[\mathcal{A}\left( f(x,y,0) \right)
    =\mathcal{A}\left(f(x,y,1) \right).\]
Now, \emph{one} of the two inputs $f(x,y,0)$ and
$f(x,y,1)$ must belong to $P$. If $f(x,y,0)\in P$, then
$\mathcal{A}\left( f(x,y,0) \right) =L(x)$,
while if $f(x,y,1) \in P$, then
$\mathcal{A}\left(  f(x,y,1) \right) =L(x)$.
Either way, then, we have learned whether $x\in L$,
and we know we have learned this.

Second, suppose
\[\mathcal{A}\left(  f(x,y,0) \right)
    \neq\mathcal{A}\left(f(x,y,1) \right).\]
Then assuming $y\in S$:
\begin{align*}
x\in L,y\in H  & \Longrightarrow\mathcal{A}
    \left( f(x,y,0)\right) =1
    \Longrightarrow\mathcal{A}\left( f(x,y,1) \right)=0,\\
x\in L,y\notin H & \Longrightarrow\mathcal{A}
    \left(f(x,y,1)\right) =1
    \Longrightarrow\mathcal{A}\left( f(x,y,0) \right)=0,\\
x\notin L,y\in H  & \Longrightarrow\mathcal{A}
    \left( f(x,y,1)\right)=0
    \Longrightarrow\mathcal{A}\left( f(x,y,0) \right)=1,\\
x\notin L,y\notin H  & \Longrightarrow\mathcal{A}
    \left(  f(x,y,0) \right)=0
    \Longrightarrow\mathcal{A}\left( f(x,y,1) \right) =1.
\end{align*}
Thus, regardless of whether $x\in L$,
we have learned whether $y\in H$, and
again we know we have learned this.

Now suppose there were an input $x\in\B^n$, such that
running $\mathcal{A}$ as above told us whether
$y\in H$ with probability more
than (say) $1/2$ over the choice of $y\sim\mathcal{D}_n$.
Then let $C_n$ be a polynomial-size circuit
that hardwires $x$, and that given an input $y\in S$:

\begin{itemize}
\item Simulates both $\mathcal{A}\left( f(x,y,0)  \right)$
    and $\mathcal{A}\left( f(x,y,1) \right)$.
\item Outputs $H(y)$ whenever it successfully
    learns the value of $H(y)$.
\item Guesses a hardwired value for $H(y)$
    (whichever of $\B$ is more probable)
    whenever it does not.
\end{itemize}

Then
\[\Pr_{y\sim\mathcal{D}_n}\left[C_n(y) =H(y)\right] 
    \geq\frac{3}{4},\]
violating the assumption that no such circuit exists.

So we conclude that for every $x\in\B^n$, we must
instead learn whether $x\in L$
with probability at least $1/2$ over the choice
of $y\sim\mathcal{D}_n$.
This, in turn, means that by simply generating
$y$'s randomly until we succeed, we can decide $L$ in
$\mathsf{P{}romiseBPP}$.
\end{proof}

Next, given a language $H\subseteq\B^{\ast}$, we say $H$
is $\mathsf{BPP}$\emph{-bi-immune} if
neither $H$ nor its complement
$\overline{H}$ has any infinite subset in $\mathsf{BPP}$.
The notion of immunity was introduced by \cite{FS74}.
Here is a useful alternative characterization:

\begin{lemma}
\label{lem:biimmune}A language $H$ is
$\mathsf{BPP}$-bi-immune if and only if
there is no infinite set $S\in\mathsf{BPP}$,
such that the promise problem
$H|_S$ is solvable in $\mathsf{PromiseBPP}$.
\end{lemma}

\begin{proof}
First, suppose $H$ is not $\mathsf{BPP}$-bi-immune,
so that either $H$ or $\overline{H}$
has an infinite subset $S\in\mathsf{BPP}$.
Then clearly, $S$ itself is an infinite set in
$\mathsf{BPP}$ such that the promise problem
$H|_S$ is trivial
(the answer is either always $0$ or always $1$).

Conversely, suppose there exists an infinite set
$S\in\mathsf{BPP}$ such that
$H|_S$ is solvable in polynomial time.
Then clearly $S\cap H$ and
$S\cap\overline{H}$ are both in $\mathsf{BPP}$,
and at least one of the two must be infinite.
So $H$ is not $\mathsf{BPP}$-bi-immune.
\end{proof}

We now suggest what, as far as we know,
is a new conjecture in quantum complexity theory.

\begin{conjecture}
\label{conj:bqpimmune}There exists a
$\mathsf{BPP}$-bi-immune language in $\mathsf{BQP}$.
\end{conjecture}

\conj{bqpimmune} is extremely strong.
Note, in particular, that
none of the \textquotedblleft standard\textquotedblright
$\mathsf{BQP}$ languages,
such as languages based on factoring or discrete log,
will be $\mathsf{BPP}$-bi-immune,
because they all have infinite special cases that
are classically recognizable and easy
(for example, the powers of $2$, in the
case of factoring).
Nevertheless, we believe \conj{bqpimmune} is plausible.
As a concrete candidate for a
$\mathsf{BPP}$-bi-immune language in $\mathsf{BQP}$,
let $g:\B^{\ast}\rightarrow\B^{\ast}$
be some strong pseudorandom generator.
Then consider the language
\[L=\left\{ x:g(x)\text{, interpreted as a positive integer,
has an odd number of distinct prime factors}\right\}. \]
We now prove \thm{immune}, restated here for convenience.


{
\renewcommand{\thetheorem}{\ref*{thm:immune}}
\begin{theorem}
Assume \conj{bqpimmune}. Then every language outside of
$\mathsf{BPP}$ is sculptable.
\end{theorem}
\addtocounter{theorem}{-1}
}

\begin{proof}
Assume by way of contradiction that $L\notin\mathsf{BPP}$
is non-sculptable.
Also, let $H$\ be a $\mathsf{BPP}$-bi-immune
language in $\mathsf{BQP}$. Then consider the set
\[ S:=\left\{  x:L(x)  =H(x)\right\}.\]
By our assumption, $S$ is a promise on which no
superpolynomial quantum
speedup is possible for $L$, and $\overline{S}$
is another such promise.
Hence, there must be a $\mathsf{BPP}$ algorithm,
call it $\mathcal{A}_S$,
that solves the promise problem $H|_S$,
which (by the definition of $S$) is
equivalent to solving $L|_S$.
And there must be another polynomial-time
classical algorithm, call it $\mathcal{A}_{\overline{S}}$,
that solves $H|_{\overline{S}}$,
which (again by the definition of $S$) is equivalent to
solving $\overline{L}|_{\overline{S}}$.

Now, given an input $x$, suppose we run both
$\mathcal{A}_S$ and $\mathcal{A}_{\overline{S}}$.
Then as in the proof of 
\thm{pad}, there are two possibilities. 
If $\mathcal{A}_S(x) = \mathcal{A}_{\overline{S}}(x)$,
then $x\in S$ implies $H(x) = \mathcal{A}_S(x)$
while $x\notin S$ implies
$H(x) =\mathcal{A}_{\overline{S}}(x)$,
so either way we have learned $H(x)$ (and we know
that we have learned it).
On the other hand, if
$\mathcal{A}_S(x) \neq\mathcal{A}_{\overline{S}}(x)$,
then $x\in S$ implies $L(x) =\mathcal{A}_S(x)$
while $x\notin S$\ implies
$L(x) =1-\mathcal{A}_{\overline{S}}(x)$.
So, merely by seeing that $\mathcal{A}_S(x)$
and $\mathcal{A}_{\overline{S}}(x)$ are different, we have
learned $L(x)$ (and we know that we have learned it).

In summary, there is a $\mathsf{BPP}$ algorithm
$\mathcal{B}$ that, for every input $x\in\B^{\ast}$,
correctly outputs either $H(x)$ or $L(x)$,
and that moreover tells us which one it output.

Now let $Q$ be the set of all $x$ such that $\mathcal{B}(x)$
outputs $H(x)$. Then there are two possibilities: if $Q$ is
finite, then $\mathcal{B}$
decides $L$ on all but finitely many inputs.
Hence $L\in\mathsf{BPP}$, contrary to assumption.
If, on the other hand, $Q$ is infinite,
then $H|_Q$ is an infinite promise problem in
$\mathsf{P{}romiseBPP}$.
So $H$ was not $\mathsf{BPP}$-bi-immune, again
contrary to assumption.
\end{proof}

In \thm{pad} and \thm{immune},
there is almost nothing
specific to the complexity classes $\mathsf{BQP}$ and
$\mathsf{BPP}$, apart from some simple closure properties.
Thus, one can prove analogous sculpting
theorems for many other pairs of complexity classes.
In some cases, we do
not even need an unproved conjecture. For example, we have:

\begin{theorem}
\label{expthm}For every language $L\notin\mathsf{P}$,
there exists a promise $S$ such that $L|_S$
is solvable in exponential time, but is not solvable
in polynomial time.
\end{theorem}

\begin{proof}
The proof of \thm{immune} follows through for $\mathsf{P}$
and $\mathsf{EXP}$ instead of $\mathsf{BPP}$ and
$\mathsf{BQP}$. In addition, it is known that there is a
$\mathsf{P}$-bi-immune language in $\mathsf{EXP}$
\cite{BH77}. The desired result follows.
\end{proof}


\section{Concluding Remarks and Open Problems}

In this work, we gave a full characterization
of the class of Boolean functions $f$
that can be sculpted into a promise problem
with an exponential quantum speedup in query complexity.
We similarly characterized sculptability for
$\R_0$ vs.\ $\R$ and $\D$ vs.\ $\R_0$.
Along the way, we showed that $\Q$ is
polynomially related (indeed, \emph{quadratically} related)
to $\D$ and $\R$ for a much wider set of promise problems
than was previously known.
Finally, we studied sculpting in
\emph{computational} complexity, giving a
strong conjecture under which every language outside
$\mathsf{BPP}$ is sculptable into a
superpolynomial quantum speedup,
and a weaker conjecture under which every \emph{paddable}
language outside $\mathsf{BPP}$ is sculptable.

One might object that many of our
sculpted promise problems are somewhat artificial.
This is particularly clear in the case of paddable languages,
where (in essence) one uses the paddability to append to each
instance $x$, as a ``comment,'' an instance of a hard
$\mathsf{BQP}$ problem (such as factoring)
that is promised to have the same answer as $x$.
Even in the query complexity setting, however, one can
observe by direct analogy that
\emph{the property of being sculptable is not
closed under the removal of dummy variables}.
So for example, we saw before that the $N$-bit
$\OR$ function is not sculptable.
By contrast, observe that the function
\[ f(x_1,\ldots,x_{2N}):=\OR(x_1,\ldots,x_N)\]
\emph{is} sculptable.
This follows as an immediate consequence of \thm{main0}:
just by adding dummy variables to the $\OR$ function,
we have vastly increased the number of inputs $x$
that have large certificate complexity,
from $1$ to $2^{N}$.
However, an even simpler way to see why $f$ is sculptable,
is that we can embed (say) Simon's problem
into the variables $x_{N+1},\ldots,x_{2N}$,
and then impose the promise that
\[\OR(x_1,\ldots,x_N) =\Simon(x_{N+1},\ldots,x_{2N})\]
(in addition to the Simon promise itself).

Of course, most Boolean functions do \emph{not}
contain such dummy variables,
so the problems of sculpting them,
and deciding whether they are sculptable at all,
are much more complicated, as we saw in this paper.

Now, it might feel like ``cheating''
to sculpt a promise problem with a large
quantum/classical gap by using dummy
variables to encode a different, unrelated problem.
If so, however, that points to an interesting
direction for future research: namely, can we somehow
formalize what we mean by a ``natural''
special case of a problem,
and can we then understand which problems are
``naturally'' sculptable?

Here are some more specific open problems.

\begin{itemize}
\item Some of our inequalities
could be off by polynomial factors;
it would be nice to tighten them (or prove separations).
For example, it may be possible to improve
\thm{small} to $\Q(f)=\Omega(\sqrt{\D(f)/\log|\Dom(f)|})$,
quadratically improving the $\log|\Dom(f)|$ factor.

\item Can our results -- and specifically, \thm{tot}
-- be used to improve the relation
$\D(f)=O(\Q(f)^6)$ due to Beals et al.\ \cite{BBC+01}?

\item Can we give a characterization of the
sculptable Boolean functions in
\emph{communication} complexity -- analogous
to this paper's characterization
of sculptability in query complexity?

\item Is there any natural pair of complexity classes
$\mathcal{C}\subseteq\mathcal{D}$,
for which $\mathcal{C}$ is known or believed to be
strictly contained in $\mathcal{D}$,
and yet it is plausible that no
languages in $\mathcal{D}$ are
$\mathcal{C}$-bi-immune, and (related to that)
there exist languages $L\notin\mathcal{C}$ that
\emph{cannot} be sculpted
into a promise problem in $\mathcal{D}\setminus\mathcal{C}$?

\item One can, of course, consider sculpting for
many other pairs of computational models,
besides $\R$ vs.\ $\Q$ or $\R_0$ vs.\ $\R$
or $\D$ vs.\ $\R_0$.
One interesting case is sculpting versus certificate
complexity -- for example, $\D$ vs.\ $\C$.
What is the correct characterization there?
\end{itemize}

We make some observations on the last problem.
It's easy to see that $\D(\OR|_P)=\C(\OR|_P)$ for
any promise $P$, so sculpting $\D$ vs.\ $\C$ is not always
possible. On the other hand,
sculpting $\D$ vs.\ $\C$
is sometimes possible even when $\Hi(\C_f)$ is small.
To see this, consider
the function $f$ with
$f(x)=1$ if and only if the Hamming weight of $x$ is $1$,
and the single `1' bit occurs on the left half of the
input string. This function
can be sculpted to $\D(f|_P)=N/2$ and $\C(f|_P)=1$ by
setting $P$ to the set of inputs with Hamming weight
$1$. However, $\Hi(\C_f)=O(\log N)$ for this
function, since all inputs with Hamming weight
at least $2$ have small certificates (just display two `1'
bits).

This means something
qualitatively different happens
with $\D$ vs.\ $\C$ than what was found in this paper.

\section*{Acknowledgements}

We thank Robin Kothari for many helpful discussions.

\noindent
Supported by an NSF Waterman Award, under grant no.\ 1249349.


\bibliographystyle{alpha}

\newcommand{\etalchar}[1]{$^{#1}$}

\appendix

\section{Properties of H Indices}
\label{app:H}
\begin{lemma}\label{Hproperties}
    Let $g:\B^n\to[0,\infty)$. Define
    \[\Hi(g):=\inf\left\{ h\in[0,\infty)
        :\left| \{x\in \B^n:g(x)> h\}\right|
            \leq 2^h\right\} .\]
    Then
    \begin{enumerate}
    \item $\Hi(g)\in[0,n]$
    \item $\Hi(g)\leq \max_x g(x)$
    \item The number of $x\in\B^n$ for which
        $g(x)>\Hi(g)$ is at most $2^{\Hi(g)}$
        (equivalently, the infimum in the definition
        of $\Hi(g)$ is actually a minimum)
    \item If $g^\prime:\B^n\to[0,\infty)$ is
        such that $g(x)\leq g^\prime(x)$
        for all $x\in\B^n$,
        then $\Hi(g)\leq\Hi(g^\prime)$
    \item If $\alpha:[0,\infty)\to[0,\infty)$ is
        an increasing function, then
        $\Hi(\alpha\circ g)\leq
            \max\{\Hi(g),\alpha(\Hi(g))\}$.
    \item There are at least $2^{\Hi(g)}$ inputs $x\in\B^n$
        with $g(x)\geq\Hi(g).$
    \end{enumerate}
\end{lemma}

\begin{proof}
Let $S_g(h)=\{x\in\B^n:g(x)>h\}$
and let $H_g=\{h\in[0,\infty):|S_g(h)|\leq 2^h\}$.
Then $\Hi(g)=\inf H_g$.
Part $1$ follows from noticing that for all $h$,
$S_g(h)\subseteq \B^n$, so $|S_g(h)|\leq 2^n$,
whence $n\in H_g$. Part $2$ follows from
noticing that $S_g(\max_x g(x))$ is empty,
so $\max_x g(x)\in H_g$.

To show $3$, we show that $H_g$ contains its infimum.
Consider an infinite
decreasing sequence $h_1,h_2,\ldots\in H_g$
that converges to $\Hi(g)$.
Then the sequence
$|S_g(h_1)|,|S_g(h_2)|,\ldots$ is a non-decreasing
sequence of integers which is bounded above by $2^n$.
In addition, $S_g(h_i)\subseteq S_g(h_{i+1})$ for all $i$.
It follows that there is some $\ell$ such that
$S_g(h_i)=S_g(h_{\ell})$ for all $i\geq\ell$.
For each $x\in S_g(h_{\ell})$,
we have $g(x)>h_\ell>\Hi(g)$,
and for each $x\notin S_g(h_{\ell})$,
we have $g(x)\leq h_i$ for all $i$.
It follows that $g(x)\leq\Hi(g)$
for each $x\notin S_g(h_{\ell})$, so
$S_g(\Hi(g))=\{x\in\B^n:g(x)>\Hi(g)\}=S_g(h_i)$
for all $i\geq \ell$.
Finally, since $h_i\in H_g$ for all $i$, we have
$|S_g(\Hi(g))|=|S_g(h_i)|\leq 2^{h_i}$
for all $i\geq\ell$. From this it follows that
$|S_g(\Hi(g))|\leq \lim_{i\to\infty} 2^{h_i}=2^{\Hi(g)}$,
so $\Hi(g)\in H_g$.

We now show $4$.
If $g^\prime$ is point-wise greater or equal to $g$, then
$S_g(\Hi(g^\prime))\subseteq S_{g^\prime}(\Hi(g^\prime))$.
Since $\Hi(g^\prime)\subseteq H_g$, we have
$|S_{g^\prime}(\Hi(g^\prime))|\leq 2^{\Hi(g^\prime)}$,
so $|S_g(\Hi(g^\prime))|\leq 2^{\Hi(g^\prime)}$.
Thus $\Hi(g^\prime)\in H_g$, so
$\Hi(g)=\inf H_g\leq \Hi(g^\prime)$.

We prove $5$.
Let $\alpha$ be an increasing function.
We have
\[S_g(\Hi(g))=\{x\in \B^n:g(x)>\Hi(g)\}
    =\{x\in \B^n:\alpha\circ g(x)>\alpha(\Hi(g))\}
    =S_{\alpha\circ g}(\alpha(\Hi(g)).\]
Thus
\[|S_{\alpha\circ g}(\max\{\Hi(g),\alpha(\Hi(g))\})|
    \leq|S_{\alpha\circ g}(\alpha(\Hi(g))|
    =|S_g(\Hi(g))|\leq 2^{\Hi(g)}
    \leq 2^{\max\{\Hi(g),\alpha(\Hi(g))\}}\]
so $\max\{\Hi(g),\alpha(\Hi(g))\}\in H_{\alpha\circ g}$.
Hence
$\Hi(\alpha\circ g)\leq\max\{\Hi(g),\alpha(\Hi(g))\}$.

Finally, we show $6$. If it was false,
there would be less than $2^{\Hi(g)}$ inputs
with $g(x)\geq\Hi(g)$. Thus there is some $\epsilon>0$
such that there are less than $2^{\Hi(g)-\epsilon}$ inputs
with $g(x)\geq\Hi(g)>\Hi(g)-\epsilon$. But this implies
$\Hi(g)-\epsilon\geq\Hi(g)$, a contradiction.

\end{proof}

\section{Proof of \lem{sauer}}
\label{app:sauer}

\sauer*

\begin{proof}
Let $d$ be the size of the largest set
that is shattered by $S$. Then the
Sauer-Shelah lemma \cite{Sau72}
states
\[|S|\leq\sum_{i=0}^d{N\choose i}.\]
A well-known bound states
\[\sum_{i=0}^d{N\choose i}\leq
2^{\mathcal{H}(d/N)N},\]
where $\mathcal{H}(d/N)$ is the
binary entropy of $d/N$. Then
\[\log_2 |S|\leq \mathcal{H}(d/N)N
=d\log_2(N/d)+(N-d)\log_2(1+d/(N-d))\]
\[\leq d\log_2(N/d)+d\log_2 e
=d\log_2 N-d\log_2 (d/e)
\leq d\log_2 N\]
(if $d\geq e$). Thus
\[d\geq \frac{\log |S|}{\log N}\]
unless $d\leq 2$.

The Sauer-Shelah lemma implies
$|S|\leq 1$ when $d=0$
and $|S|\leq N^2$ when $d=2$
(assuming $N\geq 2$).
The only problematic case is $d=1$ and $|S|=N+1$.
Thus, in all cases, we have $d\geq\log|S|/\log(N+1)$,
as desired.
\end{proof}

\end{document}